\documentclass[english,journal]{IEEEtran}

\usepackage{xcolor}
\usepackage{bm}
\usepackage{stfloats}
\usepackage{mathtools}
\usepackage{enumitem}
\usepackage{siunitx}
\usepackage{algorithm}
\usepackage{algorithmic}
\usepackage[style=ieee, doi=false, isbn=false]{biblatex}
\usepackage{makecell}
\usepackage{enumerate}
\usepackage{siunitx}

\usepackage{graphicx}
\usepackage{url}      
\usepackage{amsmath}  
\usepackage{amsfonts,amssymb}
\usepackage{microtype}
\usepackage{babel}
\usepackage{csquotes}
\usepackage{comment}

\newtheorem{lemma}{{Lemma}}
\newtheorem{definition}{{Definition}}
\newtheorem{proposition}{{Proposition}}

\newtheorem{remark}{{Remark}}
\newtheorem{corollary}{{Corollary}}
\newenvironment{proof}{{\noindent\it Proof:}}{\hfill $\square$\par}

\definecolor{purple}{RGB}{128,0,128}
\newcommand{\revis}{\textcolor{black}}
\newcommand{\rev}{\textcolor{black}}
\newcommand{\ch}{\textcolor{black}}

\newcommand{\re}{\textcolor{black}}

\let\hat\widehat
\let\tilde\widetilde

\begin{document}

\title{{Channel State Information-Free Location-Privacy Enhancement: Fake Path Injection}}

\author{Jianxiu~Li,~\IEEEmembership{Graduate Student Member,~IEEE,}
        and~Urbashi~Mitra,~\IEEEmembership{Fellow,~IEEE}
\thanks{J. Li and U. Mitra are with the Department of Electrical and Computer Engineering, University of Southern California, Los Angeles, CA 90089, USA (e-mails: {jianxiul, ubli}@usc.edu).}
\thanks{{This paper was presented in part at the 2023 IEEE International Conference on Acoustics, Speech and Signal Processing (ICASSP 2023)~\cite{LiICASSP23}}. This work has been funded largely by the USC + Amazon Center on Secure and Trusted Machine Learning as well as in part by one or more of the following: NSF CCF-1817200, DOE DE-SC0021417, Swedish Research Council 2018-04359, NSF CCF-2008927, NSF CCF-2200221, NSF RINGS-2148313, NSF CIF-2311653, ONR 503400-78050, and ONR N00014-15-1-2550.}
}

\markboth{}
{LI AND MITRA: CHANNEL STATE INFORMATION-FREE LOCATION-PRIVACY ENHANCEMENT: FAKE PATH INJECTION}

\maketitle

\begin{abstract}
In this paper, a channel state information (CSI)-free, \re{\textit{fake path}} injection (FPI) scheme is proposed for location-privacy preservation. 
By leveraging the \textit{geometrical feasibility} of the fake paths, under mild conditions, it can be proved that the illegitimate device cannot distinguish between a fake and true path, thus degrading the illegitimate devices' ability to localize. Two closed-form, lower bounds on the illegitimate devices' estimation error are derived via the analysis of the Fisher information of the location-relevant channel parameters, thus characterizing the enhanced location privacy. A transmit \revis{precoder} is proposed, which efficiently injects the virtual fake paths. The intended device receives the two parameters of the \revis{precoder} design over a secure channel in order to enable localization. The impact of leaking the \revis{precoder} structure and the associated localization leakage are analyzed. Theoretical analyses are verified via simulation. Numerical results show that a $20$dB degradation of the illegitimate devices' localization accuracy can be achieved \rev{and validate the efficacy of the proposed \re{FPI} versus using \textit{unstructured} Gaussian noise or a CSI-dependent beamforming strategy.}
\end{abstract}

\begin{IEEEkeywords} 
Location privacy, localization, \re{fake path injection}, channel state information-free, \revis{precoding}.
\end{IEEEkeywords}

\IEEEpeerreviewmaketitle

\section{Introduction}
\label{sec:intro}
\IEEEPARstart{W}{ith} the wide deployment of the Internet-of-Things, location of user equipment (UE) is becoming an important commodity. The proliferation of wearables and the associated location-based services necessitate location-privacy preservation.  Unfortunately, most existing works on localization, such as \cite{Shahmansoori,Zhou,Li}, focus on how to leverage the features of the wireless signals to improve estimation accuracy without considering location privacy. These wireless signals that have the confidential location information embedded in them are easily exposed to the risk of eavesdropping, {\it i.e.}, illegitimate devices can determine the locations of the UE when the UEs request the location-based services from legitimate devices. Even worse, more private information, such as age, gender, and personal preference, can be snooped and inferred from the leaked location information as well \cite{Ayyalasomayajula}. Hence, it is critical to develop new strategies to limit location-privacy leakage (LPL).


To combat eavesdropping, location-privacy preservation {has predominantly been} studied at the application and network layers \cite{Yu, Tomasin, Schmitt}. To {the best of our knowledge}, enabling physical-layer signals for protecting the {\bf location privacy} of UEs {has not been} fully investigated. Due to the nature of the propagation of electromagnetic waves, as long as the illegitimate devices can listen to the wireless channel and leverage the received signals for localization, the location privacy of the UE is inevitably jeopardized. Traditionally, to preserve the location privacy at the physical layer, the statistics of the channels, or the actual channels, have been exploited to prevent the illegitimate devices from easily extracting the location-relevant information from the wireless signals \cite{Goel,Tomasin2,Checa,Ayyalasomayajula}.  Example strategies include artificial noise injection \cite{Goel,Tomasin2} and beamforming design \cite{Checa,Ayyalasomayajula}. 
More precisely, to make the malicious localization harder, via artificial noise, the received signal-to-noise ratio (SNR) can be decreased for illegitimate devices; to maintain the legitimate localization accuracy, actual noise realizations have to be shared with the legitimate devices or the channel state information (CSI) is known to the UE so as to inject the artificial noise into the {\it null space} of the legitimate channel \cite{Goel,Tomasin2}. On the other hand, if Alice has the prior knowledge of the CSI for the illegitimate channel, the transmit beamformer can be specifically designed to hide the key location-relevant information \cite{Checa,Ayyalasomayajula}. In particular, by exploiting the CSI to design a transmit beamformer, \cite{Ayyalasomayajula} misleads the illegitimate device into incorrectly treating one non-line-of-sight (NLOS) path as the line-of-sight (LOS) path such that UE’s position cannot be accurately inferred, where the obfuscation is achieved by adding a delay for the transmission over the LOS path. However, all these existing physical-layer security designs \cite{Goel,Tomasin2,Checa,Ayyalasomayajula} highly rely on accurate or even perfect CSI and thus are costly.

In the sequel, motivated by the analysis of the structure of the channel model {in} \cite{Beygi, Elnakeeb, Li,Li2}, we propose a \re{\textit{fake path} injection (FPI)}-aided {location-privacy enhancement} scheme to reduce location-privacy leakage to illegitimate devices. A provably, statistically hard estimation problem is created, while providing localization guarantees to legitimate devices. In contrast to \cite{Goel, Checa,  Tomasin2,Ayyalasomayajula}, our design does not require CSI, which avoids extra channel estimation and reduces the overhead for resource-limited UEs.

The main contributions of this paper are:
\begin{itemize}
\item[1)] A general framework for \re{FPI}-aided location-privacy enhancement is proposed without CSI, where the intrinsic structure of the channel is explicitly exploited to design the \re{fake paths}.
\item[2)] The identifiability of the designed \re{fake paths} is investigated. The analysis indicates that the illegitimate devices cannot provably distinguish and remove the \re{fake paths}.
\item[3)]  In the presence of the proposed \re{fake paths}, two closed-form, lower bounds on the estimation error are derived for the illegitimate devices, suggesting appropriate values of key design parameters to enhance location privacy.
\item[4)] A \revis{precoding} strategy is designed to efficiently inject the \re{fake paths}, where LPL is mitigated, while the securely {\it shared information} (a secret key~\cite{Via,Schaefer,BashSecretKey,ZhangQ,Somalatha}) is characterized to maintain authorized devices' localization accuracy. Only two parameters need to be shared with the legitimate device. The potential leakage of the \revis{precoder} structure is also studied for the robustness of the proposed scheme.
\item[5)] Theoretical analyses are validated via simulation. \ch{As compared with the localization accuracy of legitimate devices,} the proposed method can {contribute to} $20$dB localization accuracy degradation for illegitimate devices even when the LOS path exists. If the structure of the designed \revis{precoder} is unfortunately leaked, there is still around a $4$dB accuracy degradation for the illegitimate device. {This degradation can be increased by creating more complex \revis{precoders}.}
\item[6)] For low SNRs, the proposed scheme is numerically shown to be more effective for location-privacy preservation, than the injection of \textit{unstructured} Gaussian noise, which requires the sharing of entire noise realizations. \rev{The accuracy degradation achieved by the {\bf CSI-free} FPI is also comparable to that of a {\bf CSI-dependent} beamforming design \cite{Ayyalasomayajula}.}
\end{itemize}

The present work completes our previous work~\cite{LiICASSP23} with an analysis of  identifiability of the designed \re{fake paths}, which ensures that the \re{fake paths} are challenging for the illegitimate devices to remove. In addition, considering the injected \re{fake paths}, we derive two closed-form, lower bounds on the estimation error based on the Fisher information of the location-relevant channel parameters. 
{Our work is different from the analysis of the stability of the super-resolution problem \cite{Ankur,LiTIT,candes2014towards,LIACHA} with the goal of guaranteeing high estimation accuracy with a certain statistical resolution limit at high SNR. In contrast, our lower bounds are to validate that estimation accuracy can be efficiently degraded  for an eavesdropper, if the injected fake paths are close to the true paths in terms of the location-relevant parameters. We note that,  in the presence of noise, the stability of the Fisher information matrix (FIM) of the line spectral estimation problem has been studied  (see \textit{e.g.}  \cite{MaximeISIT}) for a new algorithm-free resolution limit; our analysis is tailored to location-privacy preserving design with multi-dimensional signals.} \rev{After the conference version of this work \cite{LiICASSP23} was published, FPI was further investigated for a SIMO system in \cite{FPISIMO} and a FPI-based deceptive jamming design was examined in \cite{yildirim2024deceptive}. }

The rest of this paper is organized as follows. Section~\ref{sec:signal} introduces the signal model adopted for localization. Section~\ref{sec:framework} presents the proposed CSI-free location-privacy enhancement framework, with the \re{fake paths} designed according to the intrinsic channel structure. In Section~\ref{sec:iden}, the identifiability of the injected \re{fake paths} is investigated to show that such \re{fake paths} cannot be removed. In the presence of the \re{fake paths}, two closed-form, lower bounds on the estimation error are derived and analyzed in Section~\ref{sec:crlb}, proving the efficacy of the proposed scheme. Section \ref{sec:method} proposes a transmit \revis{precoder} to practically inject the designed \re{fake paths}. The \revis{precoder} design does not rely on \ch{the} CSI. Numerical results are provided in Section~\ref{sec:sim} to validate the theoretical analyses and highlight the performance degradation caused by the \re{FPI} to the eavesdropper. Conclusions are drawn in Section~\ref{sec:con}. Appendices \ref{sec:proofiden}, \ref{sec:proofconverge}, \ref{sec:prooferrorbound}, \rev{\ref{sec:proofserrorbound}}, and \ref{sec:proofsingularFIMbeamformer} provide the proofs for the key technical results.

We use the following notation. Scalars are denoted by lower-case letters $x$ and column vectors by bold letters $\bm{x}$. The $i$-th element of $\bm{x}$ is denoted by $\bm{x}[i]$. Matrices are denoted by bold capital letters $\bm X$ and $\boldsymbol{X}[i, j]$ is the ($i$, $j$)-th element of $\bm{X}$. The operators $\lfloor{x}\rfloor$, $|x|$, $\|\bm  x\|_{2}$, $\mathfrak{R}\{x\}$, $\mathfrak{I}\{x\}$, and $\operatorname{diag}(\mathcal{A})$ represent the largest integer that is less than or equal to $x$, the magnitude of $x$, the $\ell_2$ norm of $\bm x$, the real part of $x$, the imaginary part of $x$, and a diagonal matrix whose diagonal elements are given by $\mathcal{A}$, respectively. $\boldsymbol{I}_{l}$ stands for a $l\times l$ identity matrix and $\mathbb{P}(\cdot)$ is reserved for the probability of an event. 
The operators $\mathbb{E}\{\cdot\}$ denotes the expectation of a random variable. The operators $\operatorname{Rank}(\cdot)$, $\operatorname{Tr}(\cdot)$, $\operatorname{det}(\cdot)$, $(\cdot)^\mathrm{T}$, and $(\cdot)^{\mathrm{H}}$ are defined as the rank of a matrix, the trace of a matrix, the determinant of a matrix, the transpose of a matrix or vector, and the conjugate transpose of a vector or matrix, respectively. \rev{Key symbols are listed in Table \ref{tab:notations}}.

\begin{table*}
    \centering
    \caption{\rev{Key Symbols in the System and Design}}
    \rev{\begin{tabular}{|c||c|c||c|}
        \hline Notation & Meaning & Notation & Meaning \\
        \hline$\bm p$ & location of Alice & $\bm q$ & location of Bob \\
        \hline$\bm z$ & location of Eve & $\bm v_k$ & location of the $k$-th scatterer \\
        \hline$K$ & the number of NLOS paths & $N$ & the number of sub-carriers \\
        \hline$N_t$ & the number of transmit antennas & $G$ & the number of OFDM pilot signals \\
        \hline$d$ & distance between antennas & $\sigma^2$ & variance of complex Gaussian noise \\
        \hline$c$ & speed of light & $T_s$ & sampling period \\
        \hline$B$ & bandwidth & $\lambda_c$ & wavelength \\
        \hline$\varphi_c$ & central carrier frequency  & $x^{(g,n)}$ & the $g$-th symbol transmitted over the $n$-th sub-carrier \\
        \hline$\bm f^{(g,n)}$ & {\makecell[c]{beamforming vector for the $g$-th  transmission \\ over the $n$-th sub-carrier  }}  & $\bm s^{(g,n)}$ & the $g$-th pilot signal over the $n$-th sub-carrier \\
        \hline$y^{(g,n)}$ & the $g$-th signal received over the $n$-the sub-carrier  & $ w^{(g,n)}$ & {\makecell[c]{complex Gaussian noise for  \\ the $g$-th  transmission over the $n$-th sub-carrier  }}\\
        \hline$y_\text{Bob}^{(g,n)}$ & received signal for Bob  & $ w_\text{Bob}^{(g,n)}$ & complex Gaussian noise for Bob\\
        \hline$y_\text{Eve}^{(g,n)}$ & received signal for Eve  & $ w_\text{Eve}^{(g,n)}$ & complex Gaussian noise for Eve\\
        \hline$\bm h^{(n)}$ & the $n$-th sub-carrier public channel vector  & $\bm{\tilde{ h}}^{(n)}$ & the $n$-th sub-carrier fake channel vector  \\
        \hline$\bm h_{\text{Eve}}^{(n)}$ & the $n$-th sub-carrier public channel vector for Bob  & $\bm{{ h}}_{\text{Eve}}^{(n)}$ & the $n$-th sub-carrier public channel vector for Eve  \\
        \hline$\gamma_k$ & complex channel coefficient of the $k$-th path  & $\tau_k$ & time-of-arrival of the $k$-th path \\
        \hline$\theta_{\text{Tx},k}$ & angle-of-departure of the $k$-th path  & $\bm \alpha(\theta_\text{Tx})$ & steering vector with respect to $\theta_{\text{Tx}}$  \\
        \hline$\tilde{\gamma}_{\tilde{k}}$ & artificial channel coefficient of the $\tilde{k}$-th path  & $\tilde{\tau}_{\tilde{k}}$ & artificial time-of-arrival of the $\tilde{k}$-th path \\
        \hline$\tilde{\theta}_{\text{Tx},\tilde{k}}$ & artificial angle-of-departure of the $\tilde{k}$-th path  & $\tilde{K}+1$ & the number of fake paths  \\
        \hline$\xi^{(g,n)}$ &  {\makecell[c]{structured artificial noise for  \\ the $g$-th  transmission over the $n$-th sub-carrier  }} & $\bar{\delta}_\tau$, $\bar{\delta}_{\theta_{\text{TX}}}$ & parameters used for precoding design\\
        \hline $\bar{\bm\delta}$ & a vector consisting of $\bar{\delta}_\tau$ and $\bar{\delta}_{\theta_{\text{TX}}}$ &$\bm\Phi^{(n)}$&  {\makecell[c]{precoder for the transmissions over the $n$-th sub-carrier  }} \\
        \hline $\tilde{\bm s}^{(g,n)}$ & effective pilot signal constructed based on $\bm s^{(g,n)}$&$\bm{\tilde{ h}}_\text{Eve}^{(n)}$ & the $n$-th sub-carrier fake channel vector for Eve \\
        \hline
    \end{tabular}}
    
    \label{tab:notations}
\end{table*}

\vspace{-10pt}
\section{System Model}\label{sec:signal}
We consider a legitimate device (Bob) serving {a}  UE (Alice), as shown in Figure 1. The locations of Alice and Bob are denoted by $\bm p=[p_x,p_y]^{\mathrm{T}}\in \mathbb{R}^{2}$ and $\bm q=[q_x,q_y]^{\mathrm{T}}\in \mathbb{R}^{2}$, respectively. To acquire location-based services, Alice transmits pilot signals to Bob, while Bob estimates Alice's position based on the received signals and his location $\bm q$. We assume that the pilot signals are known {to} Bob. An illegitimate device (Eve) exists at location $\bm z=[z_x,z_y]^{\mathrm{T}}\in \mathbb{R}^{2}$. {Eve} also knows the {same} pilot signals and her location $\bm z$. By eavesdropping {on} the channel to infer $\bm p$, Eve jeopardizes Alice's location privacy.

\begin{figure}[t]
    \centering
    \includegraphics[width=0.49\textwidth]{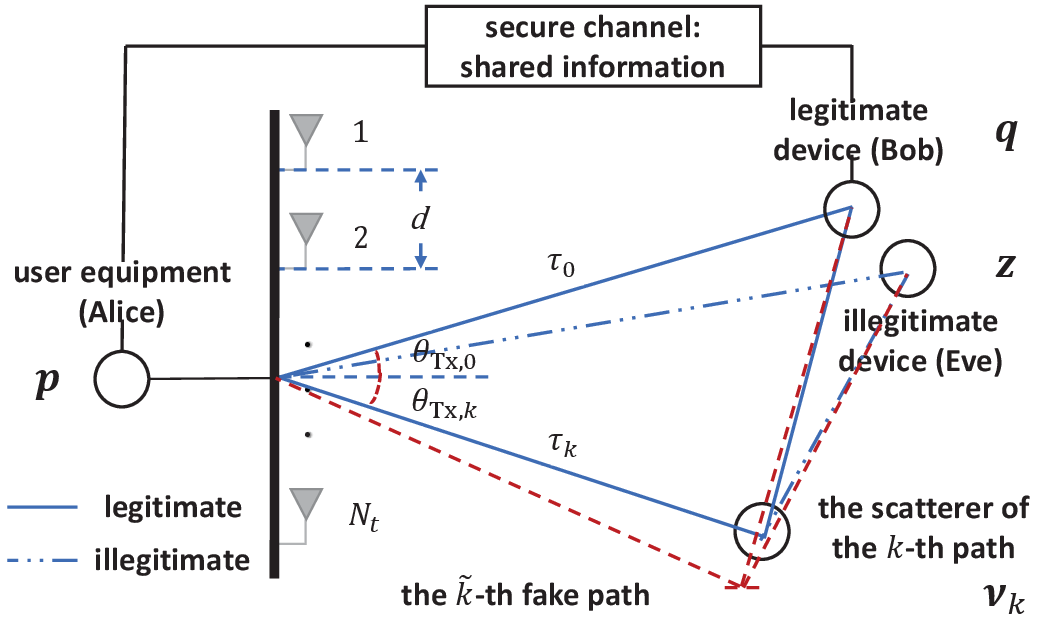}\vspace{-10pt}
    \caption{{ System model.}}\vspace{-15pt}
    \label{systemmodel}
\end{figure}

Without of loss of generality, we adopt the millimeter-wave (mmWave) multiple-input-single-output (MISO) orthogonal frequency-division multiplexing (OFDM) channel model of \cite{Fascista} for transmissions, where Alice is equipped with $N_t$ antennas, while both Bob and Eve have a single antenna\footnote{{The location-privacy enhancement framework proposed in Section \ref{sec:framework} can be extended for the single-input-multiple-output (SIMO) systems and  multiple-input-multiple-output (MIMO) systems.}}. Assume that $K+1$ paths, {\it i.e.,} one LOS path and $K$ NLOS paths, exist in the MISO OFDM channel, and the scatterer of the $k$-th NLOS path is located at an unknown position $\bm v_k=[v_{k,x},v_{k,y}]^{\mathrm{T}}\in \mathbb{R}^{2}$, for $k=1,2,\cdots,K$. We transmit $G$ OFDM pilot signals via $N$ sub-carriers with central carrier frequency $\varphi_c$ and bandwidth $B$. It is assumed that a narrowband channel is employed, {\it i.e.,} $B\ll \varphi_c$. Denoting by $x^{(g,n)}$ and $\bm f^{(g,n)}\in\mathbb{C}^{N_t\times 1}$ the $g$-th symbol transmitted over the $n$-th sub-carrier and the corresponding beamforming vector, respectively, we can express the $g$-th pilot signal over the $n$-th sub-carrier as
$\boldsymbol{s}^{(g,n)}\triangleq \bm f^{(g,n)}x^{(g,n)}\in\mathbb{C}^{N_t\times 1}$ and write the received signal ${y}^{(g,n)}$ as 
\begin{equation}
{y}^{(g,n)}=\boldsymbol{h}^{(n)}  \boldsymbol{s}^{(g,n)}+{w}^{(g,n)},\label{rsignal}
\end{equation}
for $n=0,1,\cdots,N-1$ and $g=1,2,\cdots,G$, where ${w}^{(g,n)}\sim \mathcal{CN}({0},\sigma^2)$ is an independent, zero-mean, complex Gaussian noise with variance $\sigma^2$, and $\bm h^{(n)}\in\mathbb{C}^{1\times N_t}$ represents the $n$-th sub-carrier public channel vector. We assume that the pilot signals transmitted over the $n$-th sub-carrier are independent and identically distributed and $\mathbb{E}\{\boldsymbol{s}^{(g,n)}(\boldsymbol{s}^{(g,n)})^{\mathrm{H}}\}= \frac{1}{N_t}\bm I_{N_t}$ holds for any $g$ and $n$. Denote by $c$, $d$, $T_s$ and $\bm{a}_L(f) \in \mathbb{C}^L$ the speed of light, the distance between antennas, the sampling period $T_s\triangleq\frac{1}{B}$, and the unit-norm Fourier vector
\(
    \bm{a}_L(\vartheta) \triangleq  \frac{1}{\sqrt{L}}\left[1, e^{-j 2\pi \vartheta}, \dots, e^{-j 2\pi (L-1)\vartheta}\right]^{\mathrm{T}}
\), respectively. The public channel vector $\bm h^{(n)}$ is defined as 
\begin{equation}
\boldsymbol{h}^{(n)}\triangleq\sqrt{N_t}\sum_{k=0}^{K}\gamma_k e^{\frac{-j 2\pi n\tau_k}{N T_{s}}}\boldsymbol{ \alpha}\left(\theta_{\mathrm{Tx},k}\right)^{\mathrm{H}},
\label{channelmatrix_subcarrier} 
\end{equation}
where $k=0$ corresponds to the LOS path, $\gamma_k$ represents the complex channel coefficient of the $k$-th path, while the steering vector $\boldsymbol{\alpha}(\theta_{\mathrm{Tx}})$ is defined as  $\boldsymbol{\alpha}(\theta_{\mathrm{Tx}}) \triangleq \bm{a}_{N_t}\left(\frac{d \sin(\theta_{\mathrm{Tx}})}{\lambda_c} \right)$ with $\lambda_c\triangleq\frac{c}{\varphi_c}$ being the wavelength. According to the geometry, the location-relevant channel parameters of each path,  {\it i.e.,} time-of-arrival (TOA) $\tau_{k}$ and angle-of-departure (AOD) $\theta_{\mathrm{Tx}, k}$, are given by 
\begin{subequations}\label{eq:geometry}
\begin{align}
\tau_{k} &=\frac{\left\|\boldsymbol{v}_0-\boldsymbol{v}_{k}\right\|_{2} +\left\|\boldsymbol{p}-\boldsymbol{v}_{k}\right\|_{2}} {c},\\
\theta_{\mathrm{Tx}, k} &=\arctan \left(\frac{v_{k,y}-p_{y}} {v_{k,x}-p_{x}}\right),  
\end{align}
\end{subequations} 
for $k=0,1,\cdots,K$, where $\bm v_0\triangleq \bm q$ (or $\bm v_0\triangleq \bm z$) holds for Bob (or Eve). In addition, we assume $\frac{\tau_k}{NT_s } \in(0,1]$ and ${\frac{d \sin \left(\theta_{\mathrm{Tx},k}\right)}{\lambda_{c}} }\in(-\frac{1}{2},\frac{1}{2}]$ as {in} \cite{Li}.

Given the pilot signals, the location of Alice can be estimated using all the received signals. Hereafter, we denote by ${y}^{(g,n)}_{\text{Bob}}$ and ${y}^{(g,n)}_{\text{Eve}}$ the signals received by Bob and Eve, respectively, which are defined according to Equation \eqref{rsignal}, while the public channels for Bob and Eve, denoted as $\boldsymbol{h}^{(n)}_{\text{Bob}}$ and $\boldsymbol{h}^{(n)}_{\text{Eve}}$, are modelled according to Equation \eqref{channelmatrix_subcarrier}, {as a function of} their locations. Note that {the} CSI is assumed to be unavailable {to Alice}. {We aim} to reduce {the} LPL to Eve with the designed \re{FPI}, providing localization guarantees to Bob with the shared information transmitted over a secure channel.

\vspace{-5pt}
\section{\re{Fake Path Injection}-Aided Location-Privacy Enhancement }\label{sec:framework} 
In this section, we present a general framework for CSI-free location-privacy enhancement with \re{the FPI}. Since it is assumed that the CSI is unknown, we inject the \re{fake paths} tailored to the intrinsic structure of the channel, \re{which can be considered as \textit{structured artificial noise}}, to effectively prevent Eve from accurately inferring Alice's position, and characterize the securely shared information to maintain Bob's localization accuracy.

According to the analysis of the atomic norm minimization based localization \cite{Li,TangTIT15}, high localization accuracy relies on super-resolution channel estimation; to ensure the quality of the estimate, the TOAs and AODs of the $K + 1$ paths need to be sufficiently separated, respectively. Inspired by \cite{Beygi, Elnakeeb, Li,Li2}, we propose the \re{FPI} to degrade the structure of the channel for location-privacy preservation, where the minimal separation for TOAs and AODs, {\it i.e.,} $\Delta_{\min}\left(\left\{\frac{\tau_k}{NT_s}\right\}\right)$ and $\Delta_{\min}\left(\left\{\frac{d \sin \left(\theta_{\mathrm{Tx},k}\right)}{\lambda_c}\right\}\right)$, is reduced, respectively\footnote{{For the atomic norm minimization based localization method in \cite{Li} with the mmWave MIMO OFDM signaling, the conditions $\Delta_{\min}\left(\left\{\frac{\tau_k}{NT_s}\right\}\right)\geq\frac{1}{\left\lfloor\frac{N-1}{8}\right\rfloor}$ and $\Delta_{\min}\left(\left\{\frac{d \sin \left(\theta_{\mathrm{Tx},k}\right)}{\lambda_c}\right\}\right)\geq\frac{1}{\left\lfloor\frac{N_t-1}{4}\right\rfloor}$ are desired.}}, with \vspace{-3pt}
\begin{equation}\label{eq:defminsep}
    \Delta_{\min}(\{\kappa_k\})\triangleq\min_{k\neq k^{\prime}}\min(|\kappa_k-\kappa_{k^{\prime}}|,1-|\kappa_k-\kappa_{k^{\prime}}|). \vspace{-3pt}
\end{equation}

To be more precise, in this framework, \re{denoting by $\tilde{K}$ an integer assumed to be greater than or equal to $K$, we can design a virtual fake channel with $\tilde{K}+1$ fake paths according to the channel structure as 
\begin{equation}
    \tilde{\bm h}^{(n)}\triangleq\sqrt{N_t}\sum_{\tilde{k}=0}^{\tilde{K}}\tilde{\gamma}_{\tilde{k}} e^{\frac{-j 2\pi n\tilde{\tau}_{\tilde{k}}}{N T_{s}}}\boldsymbol{ \alpha}\left(\tilde{\theta}_{\mathrm{Tx},\tilde{k}}\right)^{\mathrm{H}}\in\mathbb{C}^{1\times N_t},
    \label{fakechannel} 
\end{equation}
where $\tilde{\gamma}_{\tilde{k}}$, $\tilde{\tau}_{\tilde{k}}$, and $\tilde{\theta}_{\mathrm{Tx},\tilde{k}}$ can be interpreted as the artificial channel coefficient, TOA, and AOD of the $\tilde{k}$-th {\textbf{fake}} path, respectively. With the injection of these fake paths to the original mmWave MISO OFDM channel, which is also shown in Figure \ref{systemmodel}, the received signal {used} by Eve is 
\begin{equation}
\begin{aligned}
{y}^{(g,n)}_{\text{Eve}}&=\left({\bm h}^{(n)}_{\text{Eve}}+\tilde{\bm h}^{(n)}\right)\boldsymbol{s}^{(g,n)}+{w}_{\text{Eve}}^{(g,n)}\\
&={\bm h}^{(n)}_{\text{Eve}}  \boldsymbol{s}^{(g,n)}+ \xi^{(g,n)}+{w}_{\text{Eve}}^{(g,n)},
\end{aligned}\label{rsignalnEve} 
\end{equation}
where ${w}_{\text{Eve}}^{(g,n)}\sim \mathcal{CN}({0},\sigma^2)$. As seen in Equation \eqref{rsignalnEve}, the proposed FPI equivalently adds the structured artificial noise, denoted as $\xi^{(g,n)}$, to the $g$-th received signal transmitted over the $n$-th sub-carrier, with $n=0,1,\cdots,N-1$ and $g=1,2,\cdots,G$, \textit{i.e.,}}
\begin{equation}\label{eq:san}
{\xi}^{(g,n)}\triangleq\tilde{\bm h}^{(n)}\boldsymbol{s}^{(g,n)}. \vspace{-3pt}
\end{equation}
Suppose that these artificial channel parameters are designed to be individually close to the true channel parameters, the introduced fake paths heavily overlap with the true paths and the minimal separation for TOAs and AODs is thus reduced according to the definition of $\Delta_{\min}(\cdot)$ in Equation \eqref{eq:defminsep}, degrading the channel structure.

On the other hand, Bob is also affected by the \re{FPI}. To alleviate the distortion caused by the \re{FPI} for Bob, we assume Alice transmits the {shared information} $\{\{\tilde{\gamma}_{\tilde{k}}\}, \{\tilde{\tau}_{\tilde{k}}\}, \{\tilde{\theta}_{\mathrm{Tx},\tilde{k}}\}\}$\footnote{{The amount of shared information is determined by the design for the \re{FPI}; one specific design will be provided in Section \ref{sec:method}. }} to Bob over a secure channel that is inaccessible by Eve. Then, Bob can exploit the shared information to \rev{compensate for} the \re{fake paths} and the signal employed by Bob for localization is still \rev{in the form of}\footnote{Error in the shared information will reduce Bob’s localization accuracy, which is affected by quantization, channel statistics, communication protocol, {\it etc}. To show the efficacy of our design, perfect shared information is assumed; the study of such error is beyond the scope of this paper.} 
\begin{equation}
{y}^{(g,n)}_{\text{Bob}}={\bm h}^{(n)}_{\text{Bob}}  {\boldsymbol{s}}^{(g,n)}+{w}_{\text{Bob}}^{(g,n)},\label{rsignalnBob} 
\end{equation}
where ${w}_{\text{Bob}}^{(g,n)}\sim \mathcal{CN}({0},\sigma^2)$. 
In contrast to Bob who has the access to the shared information, the \re{FPI} distorts the structure of Eve's channel and thus degrades her eavesdropping ability. We note that the proposed location-privacy enhancement framework does not rely on the CSI and any specific estimation method. In Section \ref{sec:method}, to decrease the minimal separation without \ch{the} CSI, an efficient strategy for the \re{FPI} will be shown, where the virtually added fake paths are further elaborated {upon}, while the specific design of $\tilde{\gamma}_{\tilde{k}}$, $\tilde{\tau}_{\tilde{k}}$, $\tilde{\theta}_{\mathrm{Tx},\tilde{k}}$ and {the accompanying shared information} are provided.

\vspace{-5pt}
\section{Identifiability of \re{Fake Paths}}\label{sec:iden}
To enhance the location privacy, \re{FPI} is proposed in Section \ref{sec:framework} to distort the structure of Eve's channel. However, if Eve can distinguish and remove the designed \re{fake paths}, the efficacy of the proposed method would be severely degraded. Hence, the identifiability of the proposed \re{fake paths} is investigated in this section and conditions are provided to ensure that Eve cannot distinguish the \re{fake paths}.

\re{To show that Eve cannot distinguish the injected fake paths from the true paths, we introduce the following notion of the feasibility of paths.}

\begin{definition}\label{def:feasible}
Given a pair of TOA and AOD parameters, {\it i.e.,} $\{\tau,\theta_{\text{Tx}}\}$, a path characterized by $\{\tau,\theta_{\text{Tx}}\}$ is {\it geometrically feasible} if there exists a scatterer at location $\bm v$ that can be mapped to the parameters $\{\tau,\theta_{\text{Tx}}\}$ according to Equation \eqref{eq:geometry}.
\end{definition}

From Definition \ref{def:feasible}, to make the introduced fake paths indistinguishable, we need to show that they are geometrically feasible, where the required conditions are provided in the following proposition.

\begin{proposition}\label{prop:iden}
With the \rev{FPI proposed in Equation \eqref{rsignalnEve}}, Eve cannot distinguish between the added fake paths and the true paths if $c\tilde{\tau}_{\tilde{k}}\geq\|\bm z- \bm p\|_2$ holds for ${\tilde{k}}=0,1,\cdots,{\tilde{K}}$. 
\end{proposition}
\begin{proof}
See Appendix \ref{sec:proofiden}.
\end{proof}
According to Proposition \ref{prop:iden}, with appropriate choice of the design parameters, the \re{fake paths} cannot be identified and thus cannot be removed. With the \re{FPI}, Eve will incorrectly believe that there are $K+\tilde{K}+1$ NLOS paths associated with $K+\tilde{K}+1$ scatterers at positions ${\bm v}_k$ and  $\tilde{\bm v}_{\tilde{k}}$ with $k=1,2,\cdots,K$ and $\tilde{k}=0,1,\cdots,\tilde{K}$, where $\tilde{\bm v}_k$ is given in Equation \eqref{eq:tildevk} and can be evaluated by substituting Equation \eqref{eq:bk} into Equation \eqref{eq:tildevk}. The analysis of the identifiability of the proposed \re{fake paths} suggests the rationality of the proposed method. In Section \ref{sec:crlb}, the efficacy of the designed \re{FPI} will be investigated. 

\vspace{-3pt}
\section{Fisher Information of Channel Parameters with \re{Fake Path Injection}}\label{sec:crlb}
Since localization accuracy highly relies on the quality of the estimates of channel parameters, the Fisher information of the {location-relevant} channel parameters is analyzed in this section to show that \re{the FPI} results in a reduction of the minimal separation and thus effectively increases the estimation error.  {As an example, we see that \cite{Li} determines the minimal separation which is a sufficient condition for uniqueness and optimality of the proposed localization algorithm.}
To this end, we first present the exact expression for the FIM\footnote{Though similar derivations can be found in \cite{Fascista}, we show the FIM herein for consistency, which will be used for the analysis of our design.}. Then, two lower bounds on the estimation error are derived and analyzed based on {an {\it asymptotic}} FIM, suggesting appropriate values of the key design parameters. 

For simplicity, in this section, $\tilde{K}$ is assumed to be equal to $K$ and the $k$-th fake path is designed to be close to the $k$-th true path in terms of differences of the channel coefficients, TOAs and AODs. In addition, we assume the conditions in Proposition \ref{prop:iden} are satisfied for the analysis in this section. To quantify how close the $k$-th true path is to the $k$-th fake path, we denote by $\delta_{\gamma_k}$, $\delta_{\tau_k}$ and $\delta_{\theta_{\mathrm{Tx},{k}}}$ the differences for the channel coefficients, TOAs and AODs, respectively, and define them as \vspace{-8pt}
\begin{subequations}\label{eq:diff}
\begin{align}
\delta_{\gamma_k}&\triangleq\tilde{\gamma}_k-\gamma_k,\\
\delta_{\tau_k}&\triangleq\tilde{\tau}_k-\tau_k,\\
\delta_{\theta_{\mathrm{Tx},{k}}}&\triangleq\arcsin\left(\sin(\tilde{\theta}_{\mathrm{Tx},{k}})-\sin(\theta_{\mathrm{Tx},{k}})\right). \vspace{-5pt}
\end{align}
\end{subequations}
According to Equation \eqref{eq:defminsep}, if the values of $\delta_{\tau_k}$ and $\delta_{\theta_{\mathrm{Tx},{k}}}$ are sufficiently small, the minimal separation for TOAs and AODs are determined by $\delta_{\tau_k}$ and $\delta_{\theta_{\mathrm{Tx},{k}}}$, respectively. Hence, we will analyze the effects of $\delta_{\gamma_k}$, $\delta_{\tau_k}$ and $\delta_{\theta_{\mathrm{Tx},{k}}}$. We note that there are lower bounds on $|\delta_{\gamma_k}|$, $|\delta_{\tau_k}|$ and $|\delta_{\theta_{\mathrm{Tx},{k}}}|$ in practice, denoted as $\delta_{\gamma_k,\min}$, $\delta_{\tau_k,\min}$ and $\delta_{\theta_{\mathrm{Tx},{k}},\min}$, respectively, {\it i.e.,} $0\leq \delta_{\gamma_k,\min}\leq|\delta_{\gamma_k}|$, $0<\delta_{\tau_k,\min}\leq|\delta_{\tau_k}|$ and $0<\delta_{\theta_{\mathrm{Tx},{k}},\min}\leq|\delta_{\theta_{\mathrm{Tx},{k}}}|$ so that each pair of fake and true paths is still considered to be produced by two distinct scatterers.

\vspace{-5pt}
\subsection{Exact Expression for FIM}

Considering the injected fake paths, we stack the true and artificial channel parameters as  $\boldsymbol{\bar\gamma} \triangleq [\gamma_0,\gamma_1,\cdots,\gamma_K,\tilde{\gamma}_0,\tilde{\gamma}_1,\cdots,\tilde{\gamma}_K]^{\mathrm{T}}\in\mathbb{C}^{2(K+1)}$, $\boldsymbol{\bar\tau} \triangleq [\tau_0,\tau_1,\cdots,\tau_K,\tilde{\tau}_0,\tilde{\tau}_1,\cdots,\tilde{\tau}_K]^{\mathrm{T}}\in\mathbb{R}^{2(K+1)}$, and $\boldsymbol{\bar{\theta}}_{\mathrm{Tx}} \triangleq [{\theta}_{\mathrm{Tx},0},{\theta}_{\mathrm{Tx},1},\cdots,{\theta}_{\mathrm{Tx},K},\tilde{\theta}_{\mathrm{Tx},0},\tilde{\theta}_{\mathrm{Tx},1},\cdots,\tilde{\theta}_{\mathrm{Tx},K}]^{\mathrm{T}}\in\mathbb{R}^{2(K+1)}$. Denote the vector of all the channel parameters as 
\begin{equation}
    \bm \eta \triangleq \left[\boldsymbol{\bar\tau}^{\mathrm{T}},\boldsymbol{\bar{\theta}}_{\mathrm{Tx}}^{\mathrm{T}},\mathfrak{R}\{\boldsymbol{\bar\gamma}^{\mathrm{T}}\},\mathfrak{I}\{\boldsymbol{\bar\gamma}^{\mathrm{T}}\}\right]^{\mathrm{T}}\in\mathbb{R}^{8(K+1)}. 
\end{equation}
Accordingly, the FIM $\bm J^{(\bm \eta)}\in\mathbb{R}^{8(K+1)\times8(K+1)}$ is given by \cite{Scharf}  
\begin{equation}\label{eq:FIM}
    \bm J^{(\bm \eta)} = \begin{bmatrix}
            \bm J^{(\bm \eta)}_{\boldsymbol{\bar\tau},\boldsymbol{\bar\tau}} & \bm J^{(\bm \eta)}_{\boldsymbol{\bar\tau},\boldsymbol{\bar{\theta}}_{\mathrm{Tx}}} & \bm J^{(\bm \eta)}_{\boldsymbol{\bar\tau},\mathfrak{R}\{\boldsymbol{\bar\gamma}\}} & \bm J^{(\bm \eta)}_{\boldsymbol{\bar\tau},\mathfrak{I}\{\boldsymbol{\bar\gamma}\}}\\
            \bm J^{(\bm \eta)}_{\boldsymbol{\bar{\theta}}_{\mathrm{Tx}},\boldsymbol{\bar\tau}} & \bm J^{(\bm \eta)}_{\boldsymbol{\bar{\theta}}_{\mathrm{Tx}},\boldsymbol{\bar{\theta}}_{\mathrm{Tx}}} & \bm J^{(\bm \eta)}_{\boldsymbol{\bar{\theta}}_{\mathrm{Tx}},\mathfrak{R}\{\boldsymbol{\bar\gamma}\}} & \bm J^{(\bm \eta)}_{\boldsymbol{\bar{\theta}}_{\mathrm{Tx}},\mathfrak{I}\{\boldsymbol{\bar\gamma}\}}\\
            \bm J^{(\bm \eta)}_{\mathfrak{R}\{\boldsymbol{\bar\gamma}\},\boldsymbol{\bar\tau}} & \bm J^{(\bm \eta)}_{\mathfrak{R}\{\boldsymbol{\bar\gamma}\},\boldsymbol{\bar{\theta}}_{\mathrm{Tx}}} & \bm J^{(\bm \eta)}_{\mathfrak{R}\{\boldsymbol{\bar\gamma}\},\mathfrak{R}\{\boldsymbol{\bar\gamma}\}} & \bm J^{(\bm \eta)}_{\mathfrak{R}\{\boldsymbol{\bar\gamma}\},\mathfrak{I}\{\boldsymbol{\bar\gamma}\}}\\
            \bm J^{(\bm \eta)}_{\mathfrak{I}\{\boldsymbol{\bar\gamma}\},\boldsymbol{\bar\tau}} & \bm J^{(\bm \eta)}_{\mathfrak{I}\{\boldsymbol{\bar\gamma}\},\boldsymbol{\bar{\theta}}_{\mathrm{Tx}}} & \bm J^{(\bm \eta)}_{\mathfrak{I}\{\boldsymbol{\bar\gamma}\},\mathfrak{R}\{\boldsymbol{\bar\gamma}\}} & \bm J^{(\bm \eta)}_{\mathfrak{I}\{\boldsymbol{\bar\gamma}\},\mathfrak{I}\{\boldsymbol{\bar\gamma}\}}\\
            \end{bmatrix}, 
\end{equation}
where 
\begin{equation}
    \bm J^{(\bm \eta)}[r,u] = \frac{2}{\sigma^2}\sum_{n=0}^{N-1}\sum_{g=1}^{G}\mathfrak{R}\left\{\left(\frac{\partial \rev{u^{(g,n)}}}{\partial \bm\eta[r]}\right)^{\mathrm{H}}\frac{\partial \rev{u^{(g,n)}}}{\partial \bm\eta[u]}\right\},  
\end{equation}
with $r,u=0,1,\cdots,8K+7$.
Herein, we define 
\begin{equation}
    u^{(g,n)} \triangleq \sqrt{N_t} \sum_{k=0}^{2K+1}\bm{\bar\gamma}[k]e^{-j\frac{2\pi n\bm{\bar\tau}[k]}{NT_s}}\bm\alpha(\bm{\bar\theta}_{\text{Tx}}[k])^{\mathrm{H}} \bm s^{(g,n)}, 
\end{equation}
while compute Equation \eqref{eq:FIM} using  \vspace{-4pt}
\begin{subequations}\label{eq:partial}
\begin{align}
    \frac{\partial  u^{(g,n)}}{\partial \bm{\bar\tau}[k]} &= -j\frac{2\pi \sqrt{N_t} n}{NT_s}\bm{\bar\gamma}[k] e^{-j\frac{2\pi n\bm{\bar\tau}[k]}{NT_s}}\bm\alpha(\bm{\bar\theta}_{\text{Tx}}[k])^{\mathrm{H}}\bm s^{(g,n)},\label{eq:partialtau}\\
    \frac{\partial  u^{(g,n)}}{\partial \bm{\bar\theta}_{\text{Tx}}[k]} &= j\frac{2\pi \sqrt{N_t} d}{\lambda_c}\bm{\bar\gamma}[k] e^{-j\frac{2\pi n\bm{\bar\tau}[k]}{NT_s}}\cos(\bm{\bar\theta}_{\text{Tx}}[k])\nonumber\\
    &\times{ \bm\alpha(\bm{\bar\theta}_{\text{Tx}}[k])^{\mathrm{H}}}\operatorname{diag}([0,1,\cdots,N_t-1])\bm s^{(g,n)},\label{eq:partialtheta}\\
    \frac{\partial u^{(g,n)}}{\partial \mathfrak{R}\{\bm{\bar\gamma}[k]\}} &= \sqrt{N_t} e^{-j\frac{2\pi n\bm{\bar\tau}[k]}{NT_s}}\bm\alpha(\bm{\bar\theta}_{\text{Tx}}[k])^{\mathrm{H}}\bm s^{(g,n)},\label{eq:partialhr}\\
    \frac{\partial u^{(g,n)}}{\partial \mathfrak{I}\{\bm{\bar\gamma}[k]\}} &= j\sqrt{N_t}e^{-j\frac{2\pi n\bm{\bar\tau}[k]}{NT_s}}\bm\alpha(\bm{\bar\theta}_{\text{Tx}}[k])^{\mathrm{H}}\bm s^{(g,n)}, \label{eq:partialhi} \vspace{-3pt}
\end{align} 
\end{subequations}
for $k=0,1,\cdots,2K+1$. Denote by $\hat{\bm \eta}$ an unbiased estimator of $\bm\eta$. Based on the FIM in Equation \eqref{eq:FIM}, the mean squared error (MSE) of $\hat{\bm \eta}$ can be evaluated according to \cite{Scharf} \vspace{-3pt}
\begin{equation}\label{eq:exactCRLB}
    \mathbb{E}\left\{\left(\hat{\bm \eta}-{\bm \eta}\right)\left(\hat{\bm \eta}-{\bm \eta}\right)^{\mathrm{T}}\right\}\succeq{\left(\bm J^{(\bm\eta)}\right)^{-1}}, \vspace{-3pt}
\end{equation}
which is also well known as the Cram\'{e}r-Rao lower bound (CRLB). 
Denote by $\bm \phi \triangleq \left[\boldsymbol{p}^{\mathrm{T}},\boldsymbol{v}_1^{\mathrm{T}},\boldsymbol{v}_2^{\mathrm{T}},\cdots,\boldsymbol{v}_K^{\mathrm{T}},\tilde{\boldsymbol{v}}_0^{\mathrm{T}},\tilde{\boldsymbol{v}}_1^{\mathrm{T}},\cdots,\tilde{\boldsymbol{v}}_K^{\mathrm{T}},\mathfrak{R}\{\boldsymbol{\bar\gamma}^{\mathrm{T}}\},\mathfrak{I}\{\boldsymbol{\bar\gamma}^{\mathrm{T}}\}\right]^{\mathrm{T}}\in\mathbb{R}^{8(K+1)}$ a parameter vector including {the} positions of Alice and scatterers. 
The CRLB for localization can be obtained as well via the analysis of the associated FIM $\bm J^{(\bm \phi)}$ given by \vspace{-3pt}
\begin{equation}\label{eq:FIMloc}
    \bm J^{(\bm \phi)} = \bm\Pi\bm J^{(\bm \eta)}\bm\Pi^\mathrm{T}, \vspace{-3pt}
\end{equation}
where $\bm\Pi\triangleq\frac{\partial \bm\eta^{\mathrm{T}}}{\partial\bm\phi}\in\mathbb{R}^{8(K+1)\times8(K+1)}$.

\vspace{-5pt}
\subsection{Lower Bounds on Estimation Error}\label{subsec:lberror}

Due to the structure of the FIM, it is complicated to theoretically analyze the  exact FIM to determine how the design parameters $\delta_{\gamma_k}, \delta_{\tau_k}$, and $\delta_{\theta_{\mathrm{Tx},{k}}}$ affect the estimation error. 
To show the efficacy of our framework, we derive two lower bounds on the estimation error using an approximated FIM in this subsection, which are closed-form expressions associated with $\delta_{\gamma_k}, \delta_{\tau_k}$, and $\delta_{\theta_{\mathrm{Tx},{k}}}$, under certain mild conditions on these design parameters.

Since there are no assumptions on the path loss model and knowledge of the channel coefficients do not improve the localization accuracy \cite{Li}, the channel coefficients are considered as nuisance parameters and we restrict the analysis to the Fisher information of the location-relevant channel parameters for each pair of the true and fake paths, {\it i.e.,} $\bm\zeta_k\triangleq\left[\tau_k, {\theta}_{\mathrm{Tx},k}, \tilde{\tau}_k, {\tilde{\theta}}_{\mathrm{Tx},k}\right]^{\mathrm{T}}\in\mathbb{R}^{4}$ with $k=0,1,\cdots,K$. Denote by $\bm J^{(\bm\zeta_k)}\in\mathbb{R}^{4\times4}$ the exact expression for the FIM with respect to $\bm\zeta_k$. According to Equation \eqref{eq:FIM}, $\bm J^{(\bm\zeta_k)}$ can be analogously derived, but its complicated structure still hampers the theoretical analysis of the estimation error. 

To associate the design parameters $\delta_{\gamma_k}, \delta_{\tau_k}$, and $\delta_{\theta_{\mathrm{Tx},{k}}}$ with the estimation error in a closed-form expression, an asymptotic FIM $\breve{\bm J}^{(\bm\zeta_k)}\in\mathbb{R}^{4\times4}$ is studied, which is defined as  \vspace{-3pt}
\begin{equation}\label{eq:FIMloccp}
    \breve{\bm J}^{(\bm \zeta_k)} \rev{\triangleq} \frac{8\pi^2}{(\sigma NT_s)^2} \mathfrak{R}\left\{{\bm \Gamma^{(k)}}\right\}, \vspace{-3pt}
\end{equation}
where the matrix $\bm \Gamma^{(k)}\in \mathbb{R}^{4\times4}$ is provided in Equation \eqref{eq:L}, with constants $\Lambda$, $O_i$ and functions $M_i^{(k)}$, $i=1,2,\cdots,6$, defined as, $\Lambda\triangleq\frac{\lambda_c}{NT_sd}$, $O_1\triangleq \frac{N(N-1)(2N-1)}{6}$, $O_2\triangleq \frac{N(N-1)}{2}$, $O_3\triangleq \frac{N_t-1}{2}$, $O_4\triangleq \frac{(N_t-1)(2N_t-1)}{6}$, $O_5=N$, $O_6=1$, $M_1^{(k)} \triangleq \sum_{n=0}^{N-1}n^2e^{-j\frac{2\pi n \delta_{\tau_{k}}}{NT_s}}$, $M_2^{(k)} \triangleq \sum_{n=0}^{N-1}ne^{-j\frac{2\pi n \delta_{\tau_{k}}}{NT_s}}$, $M_3^{(k)} \triangleq \frac{1}{N_t}\sum_{n_t=0}^{N_t-1}n_te^{j\frac{2\pi n_t d \sin\left(\delta_{\theta_{\text{Tx},k}}\right)}{\lambda_c}}$, $M_4^{(k)} \triangleq \frac{1}{N_t}\sum_{n_t=0}^{N_t-1}n_t^2e^{j\frac{2\pi n_t d \sin\left(\delta_{\theta_{\text{Tx},k}}\right)}{\lambda_c}}$, $M_5^{(k)} \triangleq \sum_{n=0}^{N-1}e^{-j\frac{2\pi n \delta_{\tau_{k}}}{NT_s}}$, and $M_6^{(k)} \triangleq \frac{1}{N_t}\sum_{n_t=0}^{N_t-1}e^{j\frac{2\pi n_t d \sin\left(\delta_{\theta_{\text{Tx},k}}\right)}{\lambda_c}}$. 
\begin{figure*}[b]
\vspace{-5pt}
\hrulefill
\vspace{-3pt}
    \begin{equation}\label{eq:L}
        \begin{aligned}
            &\bm \Gamma^{(k)}\triangleq \\
            & \begin{bmatrix}
            {O_1O_6|\gamma_k|^2} & -\frac{O_2O_3|\gamma_k|^2\cos(\theta_{\text{Tx},k})}{\Lambda} & {M_1^{(k)}M_6^{(k)}\gamma^{\mathrm{H}}_k\tilde{\gamma}_k}  & -\frac{M_2^{(k)}M_3^{(k)}\gamma^{\mathrm{H}}_k\tilde{\gamma}_k\cos(\tilde{\theta}_{\text{Tx},k})}{\Lambda}\\
            -\frac{O_2O_3|\gamma_k|^2\cos({\theta}_{\text{Tx},k})}{\Lambda} & \frac{O_4O_5|\gamma_k|^2\cos^2(\theta_{\text{Tx},k})}{\Lambda^2}  & -\frac{M_2^{(k)}M_3^{(k)}\gamma^{\mathrm{H}}_k\tilde{\gamma}_k\cos({\theta}_{\text{Tx},k})}{\Lambda} & \frac{M_4^{(k)}M_5^{(k)}\gamma^{\mathrm{H}}_k\tilde{\gamma}_k\cos(\theta_{\text{Tx},k})\cos(\tilde{\theta}_{\text{Tx},k})}{\Lambda^2}\\
            \left({M_1^{(k)}M_6^{(k)}\gamma^{\mathrm{H}}_k\tilde{\gamma}_k}\right)^*&-\frac{\left(M_2^{(k)}M_3^{(k)}\gamma^{\mathrm{H}}_k\tilde{\gamma}_k\right)^*\cos({\theta}_{\text{Tx},k})}{\Lambda}&{ O_1O_6|\tilde\gamma_k|^2}&-\frac{O_2O_3|\tilde\gamma_k|^2\cos(\tilde{\theta}_{\text{Tx},k})}{\Lambda}\\
            -\frac{\left(M_2^{(k)}M_3^{(k)}\gamma^{\mathrm{H}}_k\tilde{\gamma}_k\right)^*\cos(\tilde{\theta}_{\text{Tx},k})}{\Lambda}& \frac{\left(M_4^{(k)}M_5^{(k)}\gamma^{\mathrm{H}}_k\tilde{\gamma}_k\right)^*\cos(\theta_{\text{Tx},k})\cos(\tilde{\theta}_{\text{Tx},k})}{\Lambda^2} & -\frac{O_2O_3|\tilde\gamma_k|^2\cos(\tilde{\theta}_{\text{Tx},k})}{\Lambda} & \frac{O_4O_5|\tilde\gamma_k|^2\cos^2(\theta_{\text{Tx},k})}{\Lambda^2} \\
            \end{bmatrix}
            \end{aligned}
    \end{equation}
\end{figure*}
We note that, as compared with $\frac{1}{G}{\bm J}^{(\bm \zeta_k)}$, the approximation error with $\breve{\bm J}^{(\bm \zeta_k)}$ is negligible when a large number of symbols are transmitted according to the following lemma. 
\begin{lemma}\label{lemma:converge}
    As $G\rightarrow\infty$, $\frac{1}{G}{\bm J}^{(\bm \zeta_k)}$ converges almost surely (a.s.) to $\breve{\bm J}^{(\bm \zeta_k)}$, {\it i.e.,} \vspace{-3pt}
\begin{equation}\label{eq:convJzeta}
\mathbb{P}\left(\lim_{G\rightarrow\infty}\frac{1}{G}{\bm J}^{(\bm \zeta_k)}=\breve{\bm J}^{(\bm \zeta_k)}\right)=1.\vspace{-3pt}
\end{equation}
\end{lemma}
\begin{proof}
    See Appendix \ref{sec:proofconverge}.
\end{proof}
We will show that, using such an asymptotic FIM as an approximation of $\frac{1}{G}\bm J^{(\bm\zeta_k)}$, the theoretical analysis of the degraded estimation accuracy 
is tractable.
\begin{proposition}\label{prop:singularbreveJzeta}
    $\operatorname{Rank}\{\breve{\bm J}^{(\bm\zeta_k)}\} \rightarrow \Omega^{(\bm\zeta_k)}$ as $\delta_{\gamma_k}, \delta_{\tau_k}, \delta_{\theta_{\mathrm{Tx},{k}}}\rightarrow 0$, where $\Omega^{(\bm\zeta_k)}$ is an integer with $\Omega^{(\bm\zeta_k)}<4$.
\end{proposition}
\begin{proof}
    Since $\lim_{\delta_{\gamma_k}\rightarrow0}\mathfrak{R}\left\{\gamma^{\mathrm{H}}_k\tilde{\gamma}_k\right\}=\lim_{\delta_{\gamma_k}\rightarrow0}\mathfrak{R}\left\{\tilde{\gamma}^{\mathrm{H}}_k{\gamma}_k\right\}=\lim_{\delta_{\gamma_k}\rightarrow0}|\tilde{\gamma}_k|^2=|\gamma_k|^2$ holds for any $k$ while $M_i^{(k)}$ is a function with respect to $\delta_{\tau_k}$ or $\delta_{\theta_{\mathrm{Tx},{k}}}$ with \vspace{-3pt}
\begin{subequations}
     \begin{align}
     \lim_{\delta_{\tau_k},\delta_{\theta_{\mathrm{Tx},{k}}}\rightarrow0}\mathfrak{R}\left\{M_i^{(k)}M_t^{(k)}\right\}=O_iO_t,\label{eq:limitM}\\
     \left|\mathfrak{R}\left\{M_i^{(k)}M_t^{(k)}\right\}\right|\leq O_iO_t,\label{eq:boundM} \vspace{-3pt}
     \end{align}
\end{subequations}
for $i,t=1,2,\cdots,6$, it can be verified that the third and fourth rows of $\bm\Gamma^{(k)}$ converge to its first and second rows, respectively, as $\delta_{\gamma_k}, \delta_{\tau_k}, \delta_{\theta_{\mathrm{Tx},{k}}}\rightarrow 0$, which concludes the proof.
\end{proof}
Proposition \ref{prop:singularbreveJzeta} indicates that $\breve{\bm J}^{(\bm\zeta_k)}$ tends to a {\it singular matrix} as $\delta_{\gamma_k}, \delta_{\tau_k}, \delta_{\theta_{\mathrm{Tx},{k}}}\rightarrow 0$, which can be exploited to characterize the asymptotic property for ${\bm J}^{(\bm\zeta_k)}$ based on Lemma \ref{lemma:converge} as well. 
By leveraging the approximated FIM $\breve{\bm J}^{(\bm\zeta_k)}$, we denote by $\hat{\bm \zeta}_{k}$ an unbiased estimator of $\rev{{\bm \zeta}_{k}}$ and bound the estimation error with a closed-form expression as follows.

\begin{proposition}\label{prop:errorbound}
For $k=0,1,\cdots,K$, if $\delta_{\gamma_k}=\delta_{\gamma_k,\min}=0$ holds,
for any real $\psi>0$, there always exists a positive integer $\mathcal{G}$ such that when $G\geq\mathcal{G}$. Then, the following lower bound on the MSE of $\rev{\hat{\bm \zeta}_{k}}$ holds with probability of $1$, \vspace{-3pt}
\begin{equation}\label{eq:errorbound}
    \begin{aligned}
    &\mathbb{E}\left\{\left(\hat{\bm \zeta}_{k}-{\bm \zeta}_{k}\right)^{\mathrm{T}}\left(\hat{\bm \zeta}_{k}-{\bm \zeta}_{k}\right)\right\}>\frac{1}{G}\left(\Xi^{(k)}-\psi\right), \vspace{-3pt}
    \end{aligned}
\end{equation}
where $\Xi^{(k)}$ is provided in Equation \eqref{eq:lowerboundXi}.
\end{proposition}
\begin{figure*}[b]
\vspace{-5pt}
\hrulefill 
\vspace{-3pt}
\begin{equation}\label{eq:lowerboundXi}
    \begin{aligned}
        &\Xi^{(k)}\\ &=\frac{\lambda_c\sigma^2}{2\pi^2d|\gamma_k|^2|\cos(\theta_{\text{Tx},k})\cos(\tilde{\theta}_{\text{Tx},k})|}\sqrt[4]{\underbrace{\frac{(NT_s)^4}{\left(\left(O_1O_6+\mathfrak{R}\left\{M_1M_6\right\}\right)\left(O_4O_5+\mathfrak{R}\left\{M_4M_5\right\}\right)-\left(O_2O_3+\mathfrak{R}\left\{M_2M_3\right\}\right)^2\right)}}_{\text{ }\Xi^{k}_1}}
        \\
        &\times \sqrt[4]{\underbrace{\frac{1}{\left(O_1O_6-\mathfrak{R}\left\{M_1M_6\right\}\right)\left(O_4O_5-\mathfrak{R}\left\{M_4M_5\right\}\right)-\left(O_2O_3-\mathfrak{R}\left\{M_2M_3\right\}\right)^2}}_{\text{ }\Xi^{k}_2}}.
    \end{aligned}
    \end{equation}
    \vspace{-5pt}
\hrulefill 
\begin{equation}\label{eq:slowerboundXi}
    \begin{aligned}
        \Psi^{(k)}=&\underbrace{\frac{\lambda_c^{\frac{3}{2}}\sigma^2}{\sqrt{2}\pi^\frac{5}{2}d^{\frac{3}{2}}|\gamma_k|^2|\cos(\theta_{\text{Tx},k})\cos(\tilde{\theta}_{\text{Tx},k})|}}_{\Psi^{(k)}_1}\underbrace{\sqrt[4]{\frac{T_s^4}{O_1O_4O_5O_6(6N_t^4-11N_t^3+21N_t^2-6N_t)}}}_{\Psi^{(k)}_2}\underbrace{\sqrt{\frac{ 1}{\sin\left(\delta_{\theta_{\text{Tx},k}}\right)}}}_{\Psi^{(k)}_3}
    \end{aligned}
    \end{equation}

\end{figure*}
\begin{proof}
See Appendix \ref{sec:prooferrorbound}.
\end{proof}

According to Equations \eqref{eq:limitM} and \eqref{eq:lowerboundXi}, $\Xi^{(k)}\rightarrow\infty$ as $\delta_{\tau_k}, \delta_{\theta_{\mathrm{Tx},{k}}}\rightarrow 0$; constrained by the boundary, {\it i.e.,} $\delta_{\tau_k,\min}$ and $ \delta_{\theta_{\mathrm{Tx},{k},\min}}$, smaller values for $\left|\delta_{\tau_k}\right|$ and $ \left|\delta_{\theta_{\mathrm{Tx},{k}}}\right|$ are preferred to degrade Eve's estimation accuracy. 
To make the effect of $\delta_{\tau_k}$ and $ \delta_{\theta_{\mathrm{Tx},{k}}}$ on the estimation error more clear, we can further bound the lower bound derived in Equation \eqref{eq:errorbound}. To this end, we note that, according to Equation \eqref{eq:limitM}, there exist two positive real numbers, denoted as $\delta_{\tau_k,\max}$ and $\delta_{\theta_{\mathrm{Tx},k},\max}$, such that for a given constant $\epsilon\triangleq\frac{2\left(O_1O_4O_5O_6-(O_2O_3)^2\right)}{O_1O_6+2O_2O_3+O_4O_5}>0$, the following inequalities hold,  \vspace{-3pt}
\begin{subequations}\label{eq:boundonMiMj}
    \begin{align}
        \left|\mathfrak{R}\left\{M_1^{(k)}M_6^{(k)}\right\}-O_1O_6\right|<\epsilon\\
        \left|\mathfrak{R}\left\{M_2^{(k)}M_3^{(k)}\right\}-O_2O_3\right|<\epsilon\\
        \left|\mathfrak{R}\left\{M_4^{(k)}M_5^{(k)}\right\}-O_4O_5\right|<\epsilon \vspace{-3pt}
    \end{align}
\end{subequations}
if $|\delta_{\theta_{\mathrm{Tx},k}}|\leq\delta_{\theta_{\mathrm{Tx},k},\max}$ and $|\delta_{\tau_k}|\leq\delta_{\tau_k,\max}$, which will be used to bound $\Xi^{(k)}$ with other additional assumptions.

\begin{corollary}\label{coro:errorbound}
    For $k=0,1,\cdots,K$, supposed that
    \begin{itemize}
    \item[A1)] $\delta_{\gamma_k}=\delta_{\gamma_k,\min}=0$,
    \item[A2)] $0<\delta_{\tau_k,\min}\leq\delta_{\tau_k}\leq \delta_{\tau_k,\max}$,
    \item[A3)] $0<\delta_{\theta_{\mathrm{Tx},{k}},\min}\leq\delta_{\theta_{\mathrm{Tx},{k}}}\leq\delta_{\theta_{\mathrm{Tx},k},\max}$,
    \item[A4)]
    $\max_{n_t=0,1,2,\cdots,N_t-1\atop n=0,1,2,\cdots,N-1}n_t\sin\left(\delta_{\theta_{\text{Tx},k}}\right)+n\Lambda\delta_{\tau_k}\leq\frac{\lambda_c}{4d}$,
    \item[A5)] $(N-1)\Lambda\delta_{\tau_{k},\max}\geq\sin\left(\delta_{\theta_{\mathrm{Tx},{k}}}\right)\geq(N-1)\Lambda \delta_{\tau_k}$,
\end{itemize}
the MSE of $\rev{\hat{\bm \zeta}_{k}}$ can be bounded as Equation \eqref{eq:errorbound}, where $\Xi^{(k)}$ is replaced with $\Psi^{(k)}$, given in Equation \eqref{eq:slowerboundXi}.
    
\end{corollary}
\begin{proof}
    See Appendix \ref{sec:proofserrorbound}.
\end{proof}

As observed in Equation \eqref{eq:slowerboundXi}, $\Psi^{(k)}$ can be decomposed into three terms, {\it i.e.,} $\Psi^{(k)}=\Psi^{(k)}_1\Psi^{(k)}_2\Psi^{(k)}_3$. Define $\operatorname{SNR}_{k}\triangleq\frac{|\gamma_k|^2}{\sigma^2}$ as the received SNR for the $k$-th true path. Several key properties of $\Psi^{(k)}$ are listed as follows.
\begin{itemize}
    \item[P1)] According to $\Psi^{(k)}_1$, the lower bound on the MSE of $\rev{{\bm \zeta}_{k}}$ is not only inversely proportional to $\operatorname{SNR}_k$ but also related to the AODs;
    \item[P2)] Coinciding with the analysis for the atomic norm minimization based method \cite{Li}, we have \vspace{-3pt}
    \begin{equation}
        \Psi^{(k)}_2=\frac{1}{\mathcal{O}(BNN^{\frac{3}{2}}_t)}, \vspace{-3pt}
    \end{equation}
    which suggests employing narrower bandwidth $B$, smaller number of transmit antennas and sub-carriers, {\it i.e.,} $N_t$, and $N$, to improve location-privacy enhancement;
    \item[P3)] The value of $\Psi^{(k)}_3$ increases as $\sqrt{\sin\left(\delta_{\theta_{\text{Tx},k}}\right)}$ decreases; supposed that the value of $\delta_{\theta_{\text{Tx},k}}$ is small enough such that $\delta_{\theta_{\text{Tx},k},\min}\leq\delta_{\theta_{\text{Tx},k}}\leq\frac{\pi}{2}$, the value of $\Psi^{(k)}_3$ monotonically decreases with respect to $\delta_{\theta_{\text{Tx},k}}$ and the largest value of $\Psi^{(k)}_3$ is achieved when $\delta_{\theta_{\text{Tx},k}}=\delta_{\theta_{\text{Tx},k},\min}$. In addition, since we set $ \delta_{\tau_k}=\frac{\sin\left(\delta_{\theta_{\mathrm{Tx},{k}}}\right)}{(N-1)\Lambda}$ according to the proof of Corollary \ref{coro:errorbound}, the value of $\delta_{\tau_k}$ is reduced as well when $\delta_{\theta_{\text{Tx},k}}$ decreases, showing that Eve’s eavesdropping ability can be efficiently degraded if the injected fake paths are close to the true paths.
\end{itemize}


\begin{remark}[Discussions on assumptions]
    Assumption A1 can be realized with a transmit \revis{precoder} that will be proposed in Section \ref{sec:method}. With respect to the upper bounds in assumptions A2 and A3, they are needed to show Equation \eqref{eq:Xi1pos}; it can be numerically shown that $\delta_{\tau_k,\max}$ and $\delta_{\theta_{\mathrm{Tx},k},\max}$ can be quite large in practice such that $\delta_{\tau_k,\min}\leq\delta_{\tau_k,\max}$ and $\delta_{\theta_{\mathrm{Tx},k},\min}\leq\delta_{\theta_{\mathrm{Tx},k},\max}$ can be easily satisfied. 
    For small values of $\delta_{\tau_k}$ and $\delta_{\theta_{\mathrm{Tx},k}}$, the assumptions A4 and A5 can be met.
\end{remark}

\begin{remark}[Orthogonal true paths] Inspired by the analysis of the NLOS paths for the single-carrier mmWave MIMO channels in \cite{Abu-Shaban,Mendrzik}, due to the low-scattering sparse nature of the mm-Wave channels, the true paths are not close to each other and we have \vspace{-3pt}
\begin{equation}\label{eq:approx}
    \frac{2}{\sigma^2}\sum_{n=0}^{N-1}\sum_{g=1}^{G}\mathfrak{R}\left\{\left(\frac{\partial u^{(g,n)}}{\partial \xi_k}\right)^{\mathrm{H}}\frac{\partial  u^{(g,n)}}{\partial \xi_{k^{\prime}}}\right\}\approx0, \quad k\neq k^{\prime}, \vspace{-3pt}
\end{equation}
for a large number of symbols and transmit antennas, where $\xi_k\in\{\tau_k,\theta_{\mathrm{Tx},k},\mathfrak{R}\{{\gamma_k}\},\mathfrak{I}\{{\gamma_k}\}\}$ and $\xi_{k^{\prime}}$ is defined similarly, with $k,k^{\prime}=0,1,\cdots,K-1$. Thus, the true paths of mmWave MISO OFDM channels are approximately orthogonal to each other, {\it i.e.,} by grouping the true channel parameters path-by-path, the associated FIM is almost a block diagonal matrix. Furthermore, the $k$-th path is designed to be close to the $k$-th true path so the estimation accuracy for $\rev{{\bm \zeta}_{k}}$ does not rely much on the uncertainties for the channel parameters of the other paths; if the channel coefficients are also assumed to be known, $\operatorname{Tr}\left({\left(\bm J^{(\bm\zeta_k)}\right)^{-1}}\right)$ is nearly the CRLB for the MSE of $\rev{\hat{\bm \zeta}_{k}}$.
\end{remark}

\begin{remark}[Degraded localization accuracy]
Since Eve cannot distinguish between the true paths and fake paths, as proved in Section \ref{sec:iden}, all the location-relevant channel parameters are estimated for localization. Hence, though it is unclear how the injected \re{fake paths} affect the estimation accuracy of the individual channel parameters from Proposition \ref{prop:errorbound} and Corollary \ref{coro:errorbound}, the derived lower bound in Equation \eqref{eq:errorbound} still indicates that Eve's localization accuracy can be effectively decreased with the proper design of \re{the FPI} according to Equation \eqref{eq:FIMloc}. \rev{We note that the CRLB analysis is to show the efficacy of the proposed FPI strategy in a relatively general sense; we are not designing localization schemes. The optimal parameter design depends on the specific estimator.}
\end{remark}

\vspace{-5pt}
\section{\revis{Precoder} Design for the \re{Fake Path Injection}}\label{sec:method}
From Section \ch{\ref{sec:crlb}}, it is clear that the  injection of fake paths will increase location privacy.  However,  it is impractical to create extra physical scatterers to generate the fake paths.  In the sequel, following the principle of the proposed framework, we design a transmit \revis{precoding} strategy that ensures the creation of fake paths that are close to the true ones, to efficiently reduce the LPL to Eve without the need for \ch{the} CSI.
\vspace{-12pt}
\subsection{Alice's \revis{Precoder}}\label{sec:precoder}
Let $\bar{\delta}_\tau$ and $\bar{\delta}_{\theta_{\text{TX}}}$ represent two parameters used for \revis{precoding} design, \rev{which are not functions of the CSI.} 
\rev{To enhance the location privacy, Alice still {employs} the mmWave MISO OFDM signaling, but designs a precoding matrix ${\bm \Phi}^{(n)}$ as\vspace{-3pt}
\begin{equation}
\rev{\begin{aligned}\label{eq:precoder}
   \bm\Phi^{(n)}\triangleq\bm I_{N_t} + \sqrt{N_t}e^{-j\frac{2\pi n \bar{\delta}_{\tau}}{NT_s}}\operatorname{diag}\left(\bm \alpha\left(\bar{\delta}_{\theta_\text{Tx}}\right)^{\mathrm{H}} \right), \vspace{-3pt}
\end{aligned}}
\end{equation}
for the pilot signals transmitted over the $n$-th sub-carrier}\footnote{\rev{Thus, additional power could be used for the proposed FPI strategy; however, the total transmit power can be maintained via  normalization.}}. By analyzing the received signals in the following subsections, we will observe that adopting the transmit \revis{precoder} in
Equation \eqref{eq:precoder} can equivalently inject $\rev{K+1}$ virtual fake paths to Eve's mmWave MISO OFDM channel\footnote{The proposed \revis{precoding} strategy can be directly extended to add $\nu \rev{(K+1)}$ virtual paths with $\nu\in\mathcal{N}^+$.}. 
\vspace{-12pt}
\subsection{Bob's Localization}\label{Bobloc}
Through the public channel ${\bm h}^{(n)}_{\text{Bob}}$, Bob receives \vspace{-3pt}
\begin{equation}
\rev{
{y}^{(g,n)}_{\text{Bob}}={\bm h}^{(n)}_{\text{Bob}}  \bm\Phi^{(n)}{\boldsymbol{f}}^{(g,n)}x^{(g,n)}+{w}_{\text{Bob}}^{(g,n)}},\label{rsignalbBob} \vspace{-3pt}
\end{equation}
for $n=0,1,\cdots,N-1$ and $g=1,2,\cdots,G$. \rev{Assume Bob has the knowledge of the structure of Alice's \revis{precoder} {in Equation \eqref{eq:precoder}, {\it i.e.,} the construction of ${\bm \Phi}^{(g,n)}$ given $\bar{\delta}_\tau$, and $\bar{\delta}_{\theta_{\text{TX}}}$.}} By leveraging the secure channel, Bob also knows {$\bar{\bm\delta}\triangleq[\bar{\delta}_\tau,\bar{\delta}_{\theta_{\text{TX}}}]^\mathrm{T}\in\mathbb{R}^2$ that is \ch{the} shared information}. \rev{Bob can straightforwardly compensate for the fake paths by constructing the effective pilot signals $\tilde{\bm s}^{(g,n)}$ based on the known pilot signal $\bm s^{(g,n)}$ and the precoding matrix $\bm\Phi$ as  follows
\begin{equation}
    \tilde{\bm s}^{(g,n)}\triangleq \bm\Phi{\bm s}^{(g,n)}. 
\end{equation} 
Then, the equivalent representation of Equation \eqref{rsignalbBob} is Equation \eqref{rsignalnBob}, {\it i.e.,} ${y}^{(g,n)}_{\text{Bob}}={\bm h}^{(n)}_{\text{Bob}}  \tilde{\boldsymbol{s}}^{(g,n)}+{w}_{\text{Bob}}^{(g,n)}$.  Thus, given  $\tilde{\bm s}^{(g,n)}$, Bob can estimate his true channel, ${\bm h}^{(n)}_{\text{Bob}}$.} 
{Note that the amount of shared information $\bar{\bm\delta}$ does not increase with respect to the number of received signal samples, {\it i.e.,} $NG$. 
As assumed in Section \ref{sec:framework}, Bob receives $\bar{\bm\delta}$ noiselessly through the secure channel.}

\vspace{-10pt}
\subsection{Eve's Localization}
\rev{Assume that the precoding matrix $\bm\Phi^{(n)}$ proposed in Equation \eqref{eq:precoder} is unknown to Eve.}
\rev{Then, the following received signal has to be {used} if Eve attempts to estimate Alice's position, \vspace{-3pt}
\begin{equation}
    {y}^{(g,n)}_{\text{Eve}}={\bm h}^{(n)}_{\text{Eve}}  \bm\Phi^{(n)}{\bm{f}}^{(g,n)}x^{(g,n)}. \vspace{-3pt}
\end{equation}
Since Eve has no access to the shared information but knows the pilot signal ${\bm{s}}^{(g,n)}={\bm{f}}^{(g,n)}x^{(g,n)}$, the channel that Eve aims to estimate is given by
\vspace{-3pt}
\begin{equation}
    \begin{aligned}
    &{\bm h}_{\text{Eve}}^{(n)}\bm\Phi^{(n)}\\
    &={\bm h}_{\text{Eve}}^{(n)}\left(\bm I_{N_t} + \sqrt{N_t}e^{-j\frac{2\pi n \bar{\delta}_{\tau}}{NT_s}}\operatorname{diag}\left(\bm \alpha\left(\bar{\delta}_{\theta_\text{Tx}}\right)^{\mathrm{H}} \right)\right)\\
    &={\bm h}_{\text{Eve}}^{(n)}+{\bm h}_{\text{Eve}}^{(n)}\sqrt{N_t}e^{-j\frac{2\pi n \bar{\delta}_{\tau}}{NT_s}}\operatorname{diag}\left(\bm \alpha\left(\bar{\delta}_{\theta_\text{Tx}}\right)^{\mathrm{H}} \right)\\
    &={\bm h}_{\text{Eve}}^{(n)}+\left(\sqrt{N_t}\sum_{k=0}^{K}\gamma_k e^{\frac{-j 2\pi n\tau_k}{N T_{s}}}\boldsymbol{ \alpha}\left(\theta_{\mathrm{Tx},k}\right)^{\mathrm{H}}\right)\\
    &\quad\quad\quad\quad\quad\quad\quad\quad\quad\times\sqrt{N_t}e^{-j\frac{2\pi n \bar{\delta}_{\tau}}{NT_s}}\operatorname{diag}\left(\bm \alpha\left(\bar{\delta}_{\theta_\text{Tx}}\right)^{\mathrm{H}} \right)\\
    &={\bm h}_{\text{Eve}}^{(n)}\\
    &\quad\quad+{N_t}\sum_{k=0}^{K}\gamma_k e^{\frac{-j 2\pi n(\tau_k+\bar{\delta}_{\tau})}{N T_{s}}}\boldsymbol{ \alpha}\left(\theta_{\mathrm{Tx},k}\right)^{\mathrm{H}}\operatorname{diag}\left(\bm \alpha\left(\bar{\delta}_{\theta_\text{Tx}}\right)^{\mathrm{H}} \right)\\
    &={\bm h}_{\text{Eve}}^{(n)}+\sqrt{N_t}\sum_{k=0}^{K}\gamma_k e^{\frac{-j 2\pi n\tilde{\tau}_k}{N T_{s}}}\boldsymbol{ \alpha}\left(\tilde{\theta}_{\mathrm{Tx},k}\right)^{\mathrm{H}}\\
    &={\bm h}_{\text{Eve}}^{(n)}+\tilde{\bm h}_{\text{Eve}}^{(n)},\label{eq:fakechannel}
\end{aligned}
\end{equation}
where $\tilde{\bm h}_{\text{Eve}}^{(n)}\triangleq\sqrt{N_t}\sum_{k=0}^{K}\gamma_k e^{\frac{-j 2\pi n\tilde{\tau}_k}{N T_{s}}}\boldsymbol{ \alpha}\left(\tilde{\theta}_{\mathrm{Tx},k}\right)^{\mathrm{H}}$ corresponds to the fake paths injected to Eve’s channel via the proposed precoding strategy in Section \ref{sec:precoder}.}
Herein, we have $\tilde{K}=K$ and the artificial channel parameters are designed as \vspace{-5pt}
\begin{subequations}\label{eq:cpbeamformer}
    \begin{align}
        \tilde{\gamma}_{{k}}&={\gamma}_{{k}},\\
        \tilde\tau_k&=\tau_k+\bar{\delta}_\tau,\label{eq:tildetau}\\
        \tilde\theta_{\text{Tx},k}&=\arcsin(\sin(\theta_{\text{Tx},k})+\sin(\bar{\delta}_{\theta_{\text{Tx}}})).\label{eq:tildethetatx} \vspace{-5pt}
    \end{align}
\end{subequations}
\rev{The proposed scheme is illustrated in Figure \ref{fig:precoderFPI}. We underscore that our analysis is very generous to Eve, as she has side-information that is unlikely to be known: the transmitted symbols, $x^{(g,n)}$ and beamforming vectors, $\bm{f}^{(g,n)}$.}
\begin{figure}[t]
\centering
\includegraphics[scale=0.475]{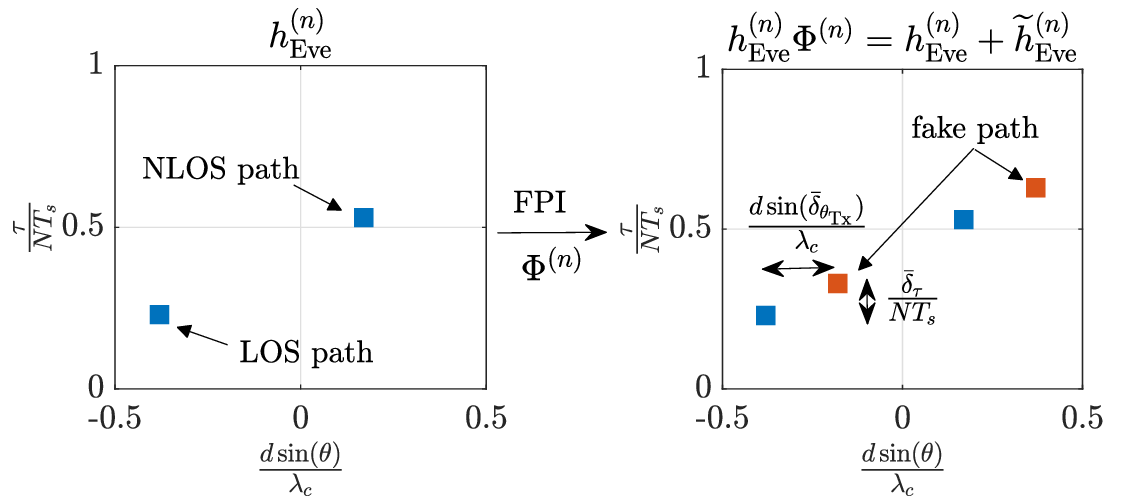}\vspace{-10pt}
\caption{\rev{An illustration of the proposed precoding design for FPI.}}\vspace{-14pt}%
\label{fig:precoderFPI}
\end{figure}

\rev{Although the injected fake paths are determined by both the true paths and the design parameters $\bar{\delta}_{\tau},\bar{\delta}_{\theta_\text{Tx}}$, the precoding matrix $\bm\Phi^{(n)}$ proposed in Equation \eqref{eq:precoder}
is not a function of the channel parameters, and fake paths can be added to the channel without the CSI.} To be more specific, if Alice uses the proposed transmit \revis{precoder}, the differences defined in Equation \eqref{eq:diff} are given by \vspace{-5pt}
\begin{subequations}\label{eq:diffbeamformer}
\begin{align}
\delta_{\gamma_0}&=\delta_{\gamma_1}=\cdots=\delta_{\gamma_K}=0,\label{eq:diffgammabeamformer}\\
\delta_{\tau_0}&=\delta_{\tau_1}=\cdots=\delta_{\tau_K}=\bar{\delta}_{\tau},\\
\delta_{\theta_{\mathrm{Tx},{0}}}&=\delta_{\theta_{\mathrm{Tx},{1}}}=\cdots=\delta_{\theta_{\mathrm{Tx},{K}}}=\bar{\delta}_{\theta_{\mathrm{Tx},}}, \vspace{-3pt}
\end{align}
\end{subequations}
Hence, given that the values of $\bar{\delta}_\tau$ and $\bar{\delta}_{\theta_{\text{Tx}}}$ are small enough, the minimal separation for TOAs and that for AODs are $\left|\frac{\bar{\delta}_\tau}{NT_s}\right|$ and $\left|\frac{d\sin(\bar{\delta}_{\theta_\text{Tx}})}{\lambda_c}\right|$, respectively, according to the definition of $\Delta_{\min}(\cdot)$ provided in Equation \eqref{eq:defminsep}, which can be efficiently decreased to degrade Eve's channel structure. \rev{We note that, as shown in Equation \eqref{eq:diffgammabeamformer}, using the proposed precoding strategy, the condition $\delta_{\gamma_k}=\delta_{\gamma_k,\min}=0$ holds for Proposition \ref{prop:errorbound} and Corollary \ref{coro:errorbound}.}
\vspace{-5pt}
\subsection{Prior Knowledge of the \revis{Precoder} Structure}\label{subsec:discussbeamformer}
Supposed that Eve does not have any prior information regarding the \revis{precoder} structure, she would employ ${\bm s}^{(g,n)}$ to infer Alice's position. For such a case, the condition to ensure the geometrical feasibility of the fake paths is provided in the following corollary based on Proposition \ref{prop:iden}.

\begin{corollary}
If Alice uses the designed \revis{precoder} in Equation \eqref{eq:precoder} with $\bar{\delta}_\tau>0$ to inject the \rev{fake paths}, Eve cannot distinguish the fake paths, even when she can perfectly estimate the TOAs and AODs of the fake paths.\label{coro:iden}
\end{corollary} 
\begin{proof}
    Given $\bar{\delta}_\tau>0$, we have  \vspace{-5pt}
    \begin{equation}
    c\tilde{\tau}_k\overset{\mathrm{(m)}}{>}c\tau_k\overset{\mathrm{(n)}}{\geq}\|\bm z-\bm p\|_2, \vspace{-5pt}
    \end{equation}
    where $\mathrm{(m)}$ follows from Equation \eqref{eq:tildetau} while $\mathrm{(n)}$ holds due to the geometry and the triangular inequality. Then, according to Proposition \ref{prop:iden}, the proof is concluded.
\end{proof}

Consequently, as analyzed in Section \ref{sec:iden}, using a proper choice of $\bar{\bm\delta}$, the proposed \re{FPI} can mislead Eve into believing that there are $2K+1$ NLOS paths while the existing paths heavily overlaps. Hence, it is unnecessary to refresh the shared information $\bar{\bm \delta}$ in the \revis{precoding} design to prevent it from being deciphered in practice if Eve does not actively attempt to snoop Alice's \revis{precoder} structure. Furthermore, as $\bar{\delta}_\tau$ and $\bar{\delta}_{\theta_{\text{Tx}}}$ \ch{become smaller}, the estimation error for all the location-relevant channel parameters tends to significantly increase based on the analysis of the associated FIM in Section \ref{subsec:lberror}; smallest values of $\bar{\delta}_\tau$ and $\bar{\delta}_{\theta_{\text{Tx}}}$ are desired according to the derived lower bound on estimation error using Equation \eqref{eq:slowerboundXi}.

On the other hand, if Alice’s \revis{precoder} structure is unfortunately leaked, Eve would learn the shared information. To be more precise, in contrast to Bob who exploits ${y}^{(g,n)}_{\text{Bob}}$ and $\tilde{\bm  s}^{(g,n)}_{\text{Bob}}$ to estimate $\left\{\left\{ \tau_k\right\},\left\{\theta_{\mathrm{Tx},{k}}\right\}, \mathfrak{R}\{{\gamma}_k\}, \mathfrak{I}\{{\gamma}_k\}\right\}$, Eve can infer $\boldsymbol{\chi}\triangleq\left\{\left\{ \tau_k\right\},\left\{\theta_{\mathrm{Tx},{k}}\right\}, \mathfrak{R}\{{\gamma}_k\}, \mathfrak{I}\{{\gamma}_k\},\bar{\delta}_\tau,\bar{\delta}_{\theta_{\text{TX}}}\right\}$ using ${y}^{(g,n)}_{\text{Eve}}$ and ${\bm  s}^{(g,n)}_{\text{Eve}}$. The corresponding FIM for the channel estimation and localization can be derived, similar to Equations \eqref{eq:FIM} and \eqref{eq:FIMloc}. 
We denote by $\boldsymbol{J}^{(\boldsymbol{\chi})}\in\mathbb{R}^{(4K+6)\times(4K+6)}$ the FIM for Eve's channel estimation with the prior knowledge of Alice's \revis{precoder} structure, whose asymptotic property is studied as follows. 
\begin{proposition}\label{prop:singularFIMbeamformer} As    $\bar\delta_{\tau_k}, \bar\delta_{\theta_{\mathrm{Tx},{k}}}\rightarrow 0$ and $G,N\rightarrow \infty$,
    $\operatorname{Rank}\{\boldsymbol{J}^{(\boldsymbol{\chi})}\}\rightarrow \Omega^{(\boldsymbol{\chi})}$ a.s., where $\Omega^{(\boldsymbol{\chi})}$ is an integer with $\Omega^{(\boldsymbol{\chi})}<4K+6$.
\end{proposition}
\begin{proof}
    See Appendix \ref{sec:proofsingularFIMbeamformer}.
\end{proof}

According to Proposition \ref{prop:singularFIMbeamformer}, $\boldsymbol{J}^{(\boldsymbol{\chi})}$ also tends to a singular matrix when    $\bar\delta_{\tau_k}, \bar\delta_{\theta_{\mathrm{Tx},{k}}}\rightarrow 0$ and $G,N\rightarrow \infty$. Thus, the FIM for Eve's localization in this case is also asymptotically rank-deficient. Based on the definition of CRLB in Equation \eqref{eq:exactCRLB}, large estimation error still can be achieved with small values of $\bar\delta_{\tau_k}, \bar\delta_{\theta_{\mathrm{Tx},{k}}}$ employed in the design of Alice's \revis{precoder}, even when the \revis{precoder} structure is snooped by Eve. However, in contrast to the case where Eve does not know the \revis{precoder} structure, the shared information $\bar{\bm \delta}$ has to be \ch{refreshed at a certain rate such that Eve cannot decipher it} though the optimal design of the refresh rate is beyond the scope of this paper. We will numerically show the impact of the \revis{precoder} structure leakage in Section \ref{sec:sim}, yet there is still a strong degradation of Eve's estimation accuracy with a proper choice of the design parameters.

\section{Simulation Results}\label{sec:sim}

In this section, we evaluate the performance of our proposed scheme with the CRLB derived in Sections \ref{sec:crlb} and \ref{sec:method}, which is not restricted to any specific estimators. First, the theoretical analyses presented in Section \ref{subsec:lberror} are numerically validated. Then, our location-privacy enhanced scheme is compared with the case where location-privacy preservation is not considered to show the degraded estimation accuracy of individual location-relevant channel parameters as well as the location. Finally, comparisons to the unstructured Gaussian noise and \rev{a CSI-dependent beamforming design \cite{Ayyalasomayajula}} are conducted to validate the efficacy of the proposed \re{FPI}, \rev{where the associated CRLB and the \textit{misspecified} Cram\'{e}r-Rao bound (MCRB) \cite{Fortunatimismatchsurvey,RichmondMCRB} are investigated, respectively.}

\subsection{Signal Parameters}
In all of the numerical results, unless otherwise stated, the system parameters $B$, $\varphi_c$, $c$,  $N_t$, $N$, $G$, $K$, and $d$ are set to $15$ MHz, $60$ GHz, $300$ m/us, $16$, $16$, $16$, $2$, and $\frac{\lambda_c}{2}$, respectively. The free-space path loss model \cite{Goldsmith} is used to {determine} channel coefficients in the simulation, while the pilot signals $\boldsymbol{s}^{(g,n)}$ are random, complex values uniformly generated on the unit circle, scaled by a factor of $\frac{1}{\sqrt{N_t}}$. The scatterers of the two NLOS paths are located at $[8.89\text{ m}, -6.05 \text{ m}]^{\mathrm{T}}$ and $[7.45 \text{ m}, 8.54 \text{ m}]^{\mathrm{T}}$, respectively, while Alice is at $[3 \text{ m},0 \text{ m}]^{\mathrm{T}}$. To make a fair comparison, Bob and Eve are placed at the same location, {\it i.e.,} $[10 \text{ m},5 \text{ m}]^{\mathrm{T}}$, and the same received signal is used for the simulation. To enhance the location privacy, $\rev{\tilde{K}+1=K+1}$ fake paths are injected via the design of the transmit \revis{precoder} according to Section \ref{sec:method}.   
In the presence of independent, zero mean, complex Gaussian noise and the proposed FPI, the received SNR is defined as $10\log_{10}\frac{\sum^{G}_{g=1}\sum^{N-1}_{n=0}\left|{\bm h}^{(n)}_{\text{Bob}}  \tilde{\boldsymbol{s}}^{(g,n)}\right|^{2}}{NG\sigma^2}$\footnote{It is also equal to $10\log_{10}\frac{\sum^{G}_{g=1}\sum^{N-1}_{n=0}\left|\left({\bm h}^{(n)}_{\text{Eve}}+\tilde{\bm h}^{(n)}\right) {\boldsymbol{s}}^{(g,n)}\right|^{2}}{NG\sigma^2}$ for a fair comparison.}. The minimal separation constraints desired by the atomic norm minimization based method \cite{Li} is denoted as $\Upsilon_{\tau}\triangleq\frac{NT_s}{\left\lfloor\frac{N-1}{4}\right\rfloor}$ and $\Upsilon_{\theta}\triangleq\arcsin\left(\frac{\lambda_c}{d\left\lfloor\frac{N_t-1}{4}\right\rfloor}\right)$, 
for the TOAs and AODs, respectively.

\begin{figure}[t]
\centering
\includegraphics[scale=0.48]{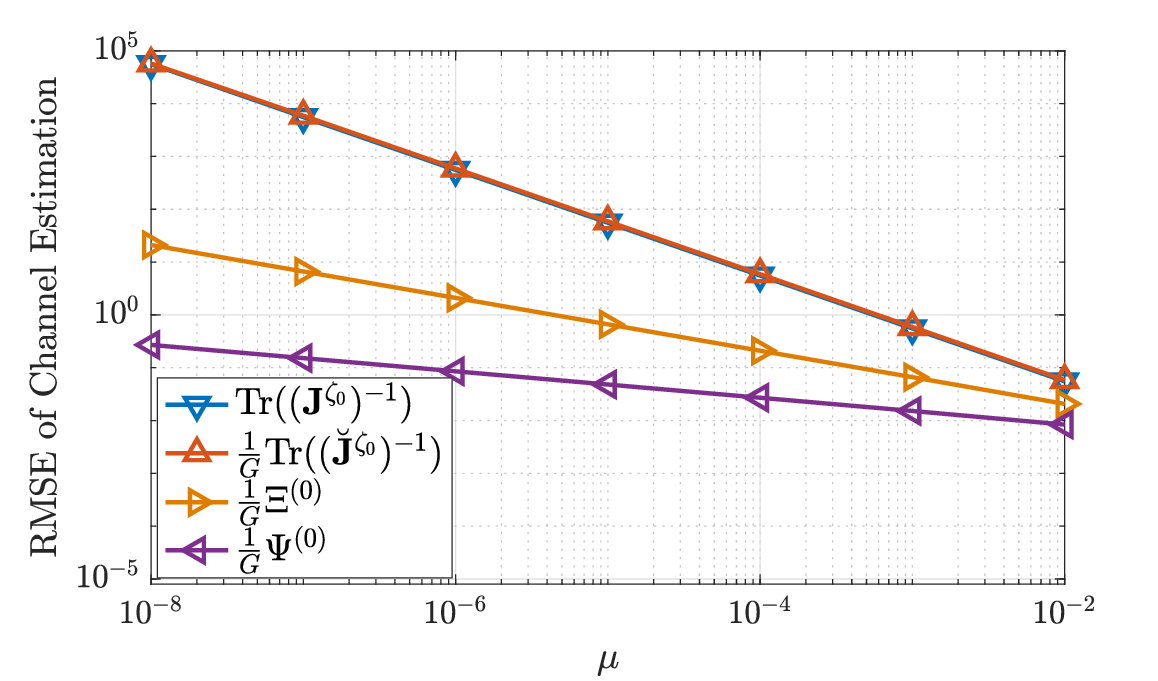}\vspace{-10pt}
\caption{{Error bounds for Eve's channel estimation.}}\vspace{-14pt}%
\label{fig:errorbound}
\end{figure}

\subsection{Validation of Theoretical Analyses in Section \ref{subsec:lberror}}\label{subsec:validation}
In terms of the root-mean-square error (RMSE) for the estimation of the LOS path and the corresponding fake path\footnote{The estimation error for the other paths can be analogously studied as well.}, Figure \ref{fig:errorbound} shows the lower bounds derived in Proposition \ref{prop:errorbound} and Corollary \ref{coro:errorbound}, where $\operatorname{Tr}\left(\left(\boldsymbol{J}^{(\boldsymbol{\zeta}_0)}\right)^{-1}\right)$ and $\rev{\frac{1}{G}}\operatorname{Tr}\left(\left(\breve{\bm J}^{(\bm \zeta_0)}\right)^{-1}\right)$ are provided as comparisons. The received SNR is set to $0$dB. To manifest the effect of $\bar\delta_{\tau}$ and $\bar\delta_{\theta_{\text{TX}}}$ on the estimation error, we set $\bar\delta_{\theta_{\text{TX}}}$ and $\bar\delta_{\tau}$ to $\mu\Upsilon_{\theta}$ and $\frac{\sin\left(\mu\Upsilon_{\theta}\right)}{(N-1)\Lambda}$, respectively, with $\mu$ being a constant. As seen in Figure \ref{fig:errorbound}, $\frac{1}{G}\Xi^{(0)}$ is a lower bound for $\operatorname{Tr}\left(\left(\boldsymbol{J}^{(\boldsymbol{\zeta}_0)}\right)^{-1}\right)$, which is further bounded by $\frac{1}{G}\Psi^{(0)}$; with the decreasing value of $\mu$, the values for $\bar\delta_{\tau}$ and $\bar\delta_{\theta_{\text{TX}}}$ decrease accordingly and the reduced distance between the truth path and fake path results in significant increases of the estimation error, though these lower bounds are not sharp, consistent with our analysis in Section \ref{sec:crlb}. In addition, from Figure \ref{fig:errorbound}, we can also observe that $\frac{1}{G}\operatorname{Tr}\left(\left(\breve{\bm J}^{(\bm \zeta_0)}\right)^{-1}\right)\approx\operatorname{Tr}\left(\left(\boldsymbol{J}^{(\boldsymbol{\zeta}_0)}\right)^{-1}\right)$, indicating the small approximation error with $\breve{\bm J}^{(\bm \zeta_0)}$ even for a realistic setting of $G$. 

Perturbed by the proposed \re{fake paths}, the CRLB for Eve’s localization is presented in Figure \ref{fig:rmseloc_delta} with different choices of $\bar\delta_\tau$ and $\bar\delta_{\theta_{\text{TX}}}$. According to Figure \ref{fig:rmseloc_delta}, coinciding with the analysis of the derived lower bounds, simultaneously decreasing the values of $\bar\delta_\tau$ and $\bar\delta_{\theta_{\text{TX}}}$ is desired to effectively degrade Eve's localization accuracy. We note that, with the choices of $\bar\delta_\tau$ and $\bar\delta_{\theta_{\text{TX}}}$ in Figure \ref{fig:rmseloc_delta}, the assumptions A4 and A5 in Corollary \ref{coro:errorbound} are not satisfied, suggesting that these sufficient conditions for the derivation of $\Psi^{(k)}$ are not necessary for the location-privacy enhancement.

\begin{figure}[t]
\centering
\includegraphics[scale=0.48]{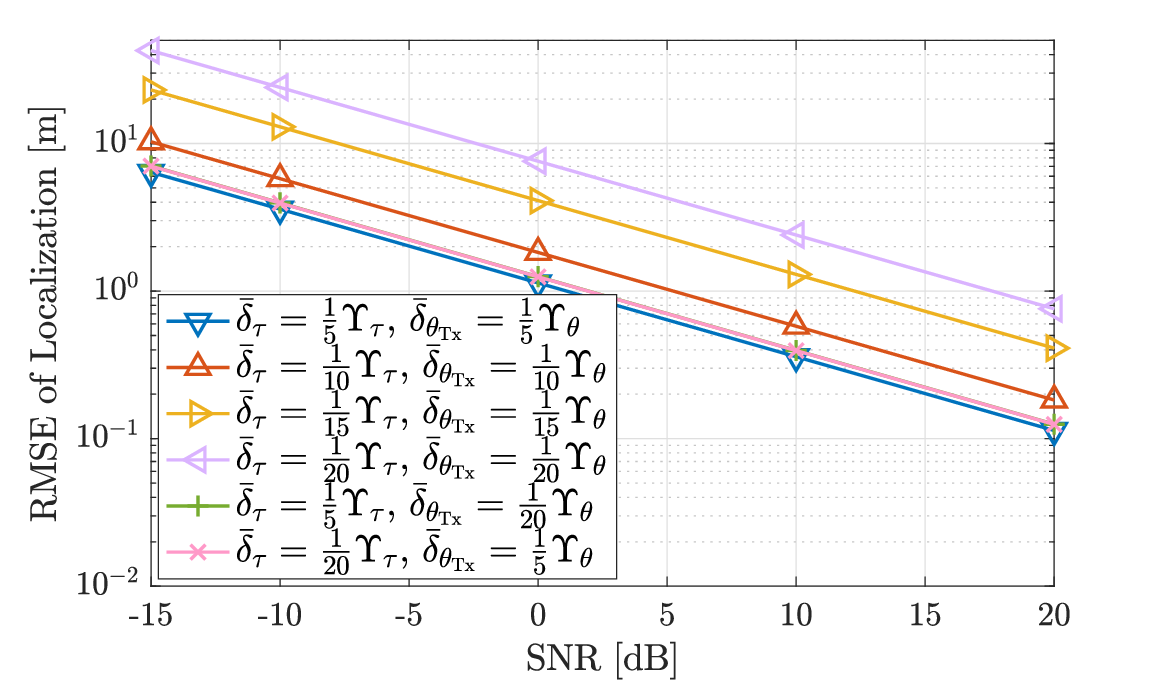}\vspace{-8pt}
\caption{{The influence of the choices of $\bar\delta_\tau$ and $\bar\delta_{\theta}$ on the $\sqrt{\text{CRLB}}$ for Eve's localization.}}\vspace{-12pt}%
\label{fig:rmseloc_delta}
\end{figure}

\subsection{Estimation Accuracy Comparison}\label{subsec:accuracycomparison}



\begin{figure}[t]
\centering \vspace{-2pt}
\includegraphics[scale=0.48]{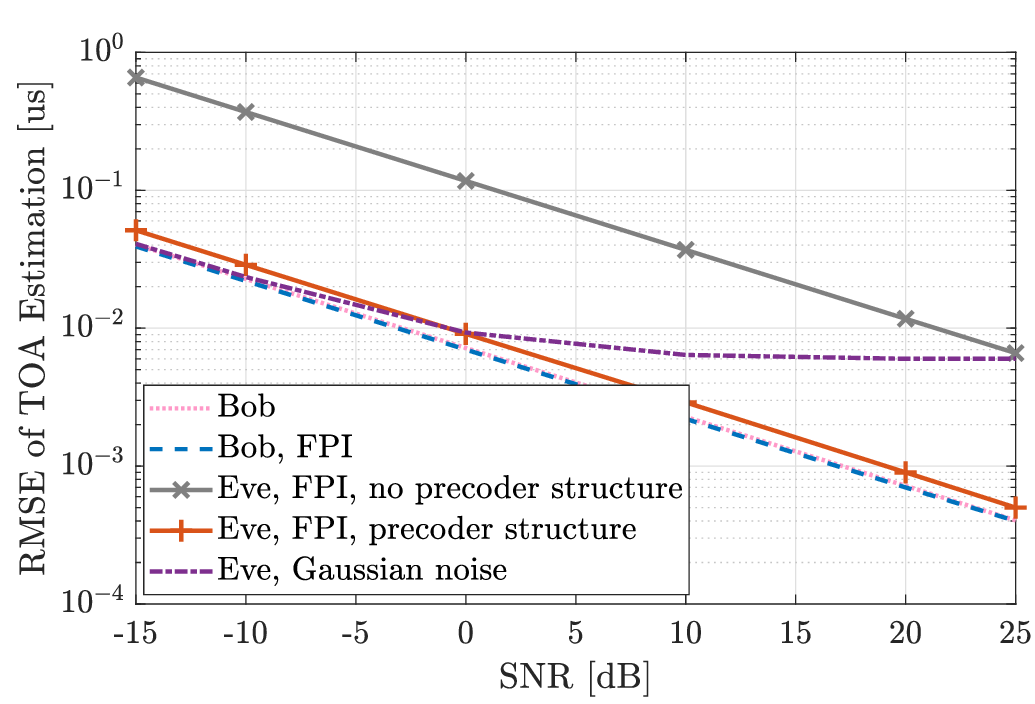}\vspace{-10pt}
\caption{{The $\sqrt{\text{CRLB}}$ for TOA estimation.}}\vspace{-12pt}%
\label{fig:crlbtoa}
\end{figure}
\begin{figure}[t]
\centering
\includegraphics[scale=0.48]{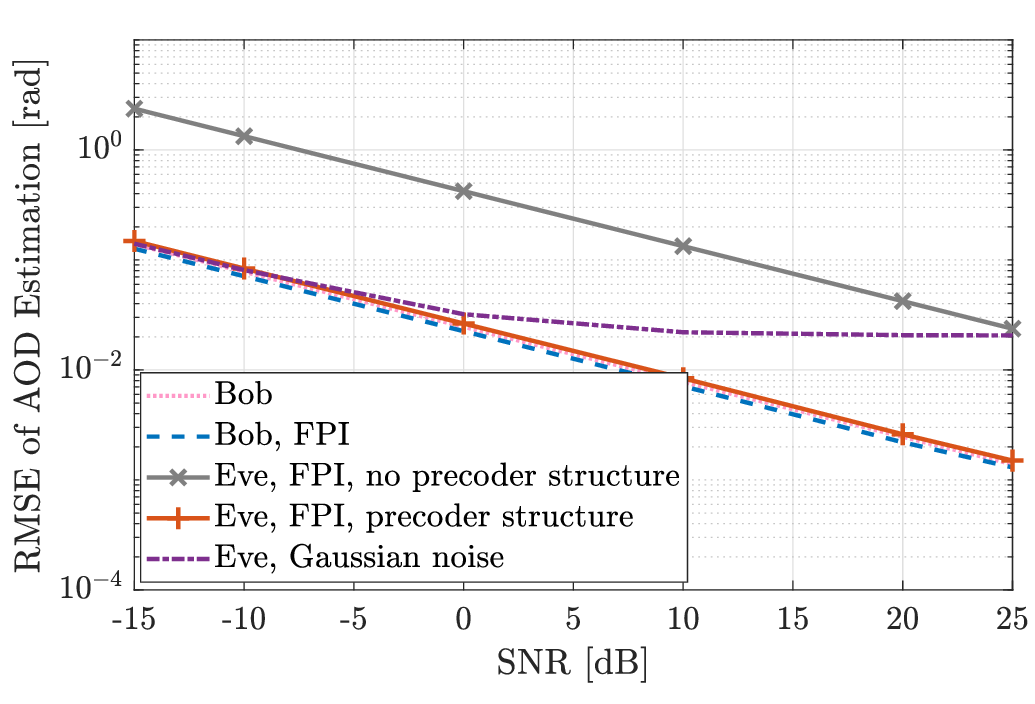}\vspace{-13pt}
\caption{{The $\sqrt{\text{CRLB}}$ for AOD estimation.}}\vspace{-8pt}%
\label{fig:crlbaod}
\end{figure}

\begin{figure}[t]
\centering
\includegraphics[scale=0.48]{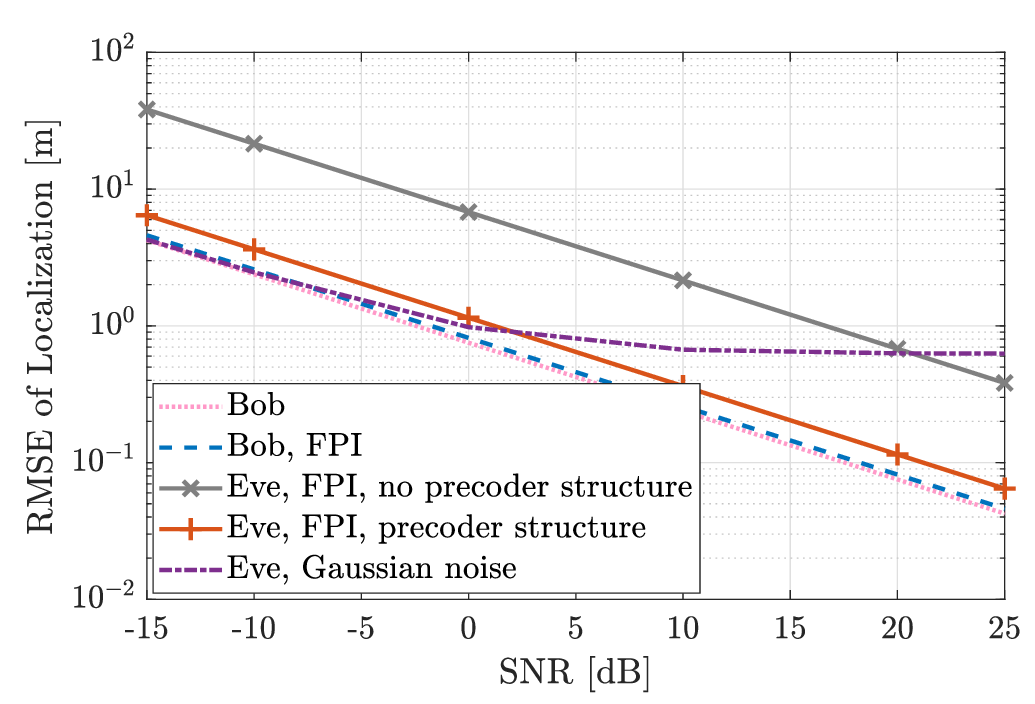}\vspace{-9pt}
\caption{{The $\sqrt{\text{CRLB}}$ for localization.}}\vspace{-16pt}%
\label{fig:crlbloc}
\end{figure}


According to Sections \ref{sec:crlb} and \ref{subsec:validation}, we assume $\delta_{\tau_k,\min}=\frac{1}{20}\Upsilon_{\tau}$ and $\delta_{\theta_{\mathrm{Tx},{k}},\min}=\frac{1}{20}\Upsilon_{\theta}$ for any $k$ and we set $\bar\delta_\tau$ and $\bar\delta_{\theta_{\text{TX}}}$ to $\delta_{\tau_k,\min}$ and $\delta_{\theta_{\mathrm{Tx},{k}},\min}$, respectively, to enhance the location privacy. The RMSE of TOA estimation and AOD estimation is shown in Figures \ref{fig:crlbtoa} and \ref{fig:crlbaod}, respectively, where the CRLB for Bob's estimation is compared with that for Eve's estimation. As observed in Figures \ref{fig:crlbtoa} and \ref{fig:crlbaod}, without the leakage of the channel structure, our proposed scheme contributes to more than $25$dB degradation with respect to the TOA and AOD estimation, by virtue of the fact that \re{the FPI} effectively distorts the structure of Eve's channel. In contrast,  Bob can compensate for the presence of the fake paths given the secure side information.

Due to the strong degradation of the quality of channel estimates, {a larger} CRLB for Eve's localization is achieved as seen in Figure \ref{fig:crlbloc}. With respect to the localization accuracy, there is {a $20$dB advantage} for Bob versus Eve using our proposed scheme. As analyzed in Section \ref{subsec:discussbeamformer}, if the structure of Alice's \revis{precoder} is unfortunately leaked to Eve, Eve can actively estimate the shared information $\bar{\bm\delta}$ to mitigate the perturbation of the \re{fake paths} while inferring Alice's position so the efficacy of our scheme is degraded. However, considering the uncertainties in the shared information, there is still around $4$dB degradation of localization accuracy for Eve versus Bob according to Figure \ref{fig:crlbloc}, indicating the robustness of our scheme with respect to the \revis{precoder} structure leakage. We note that the model order can be adjusted via changing the value of $\tilde{K}$ at a certain rate such that Eve has insufficient samples to learn the true \revis{precoder} structure and has to tackle the \re{virtually introduced fake paths}. {In addition, as shown in Figure \ref{fig:crlbloc2K}, if we inject an additional set of fake paths, Eve’s localization accuracy can be further degraded at the cost of higher transmit power. On the other hand, more side information needs to be shared with Bob to maintain his performance. Hence, there is an interesting trade-off to be investigated in the future.}

\begin{figure}[t]
\centering
\includegraphics[scale=0.48]{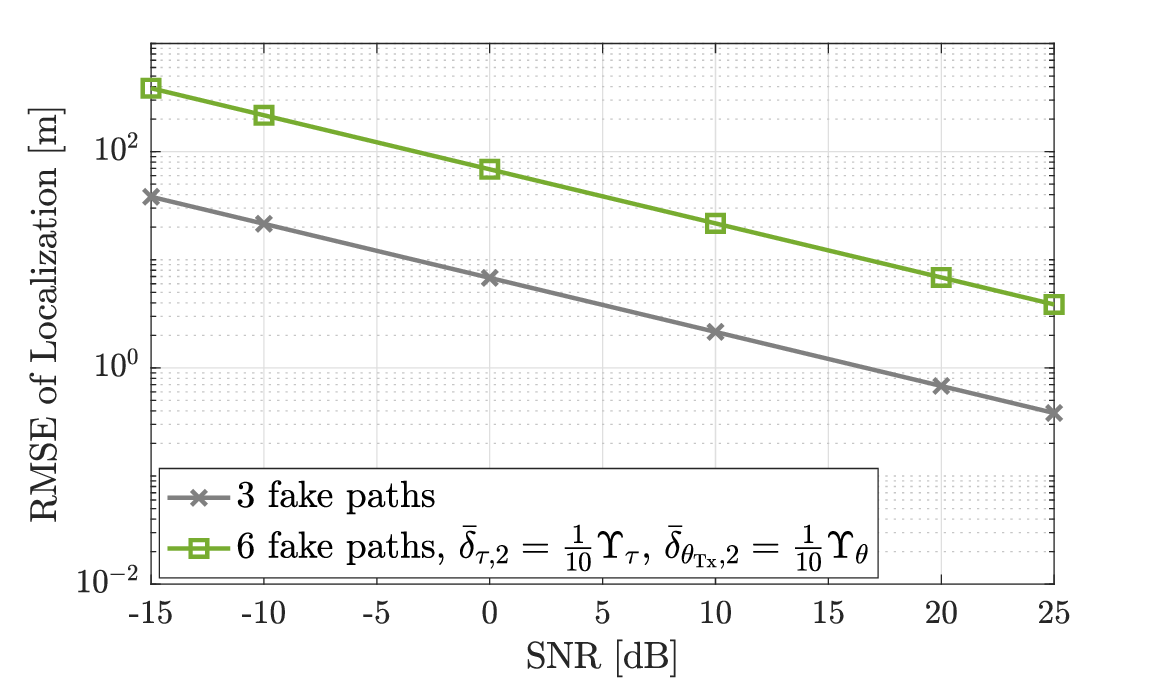}\vspace{-10pt}
\caption{{The $\sqrt{\text{CRLB}}$ for localization, 3 fake paths versus 6 fake paths, where  $\bar\delta_\tau$ and $\bar\delta_{\theta_\mathrm{Tx}}$ are still set to $\frac{1}{20}\Upsilon_{\tau}$ and $\frac{1}{20}\Upsilon_{\theta}$ for the original 3 fake paths. With respect to the additionally injected 3 fake paths and the true paths, the separations for their TOAs and AODs are determined by the design parameters $\bar\delta_{\tau,2}$ and $\bar\delta_{\theta_\mathrm{Tx},2}$ that are set to $\frac{1}{10}\Upsilon_{\tau}$ and $\frac{1}{10}\Upsilon_{\theta}$}, respectively.}\vspace{-16pt}%
\label{fig:crlbloc2K}
\end{figure}


To validate the efficacy of \re{the FPI}, CRLBs for channel estimation and localization with the injection of additional Gaussian noise \cite{Goel,Tomasin2,Xu} are also provided in Figures \ref{fig:crlbtoa}, \ref{fig:crlbaod} and \ref{fig:crlbloc} as comparisons. For relatively fair comparisons, the variance of the artificially added Gaussian noise $\tilde{w}^{(g,n)}\sim \mathcal{CN}({0},\varsigma^2)$ is set to a constant value for all the received SNRs, {\it i.e.,} $\varsigma^2\triangleq\frac{\sum^{G}_{g=1}\sum^{N-1}_{n=0}\left|\tilde{\bm 
h}^{(n)} {\boldsymbol{s}}^{(g,n)}\right|^{2}_{}}{NG}$, while the received SNR is still defined as previously stated. As observed in Figures \ref{fig:crlbtoa}, \ref{fig:crlbaod} and \ref{fig:crlbloc}, the injection of the extra Gaussian noise is ineffective to preserve location privacy at low SNRs since its constant variance is relatively small as compared with that for the Gaussian noise naturally introduced during the transmission over the wireless channel. In contrast, through the degradation of the channel structure, the proposed \re{FPI} strongly enhances location privacy. As the received SNR increases, the injection of the additional Gaussian noise can further degrade Eve's performance due to the constant variance. However, the comparison is not fully fair as the amount of side information needed by Bob to remove the artificially injected Gaussian noise without CSI is high, {\it i.e.}, the realization of $\tilde{w}^{(g,n)}$ for all $g=1,2,\cdots,G$ and $n=0,1,\cdots,N-1$; our \revis{precoding} strategy requires very little side information, {\it i.e.}, $\bar\delta_\tau$ and $\bar\delta_{\theta_{\text{TX}}}$. The exact amount of shared information transmitted over the secure channel depends on the distribution of the shared information required for the \re{FPI}, refresh rate, the quantization and coding strategy. The associated analysis is beyond the scope of this paper. 


 \begin{figure}[t]
\centering
\includegraphics[scale=0.48]{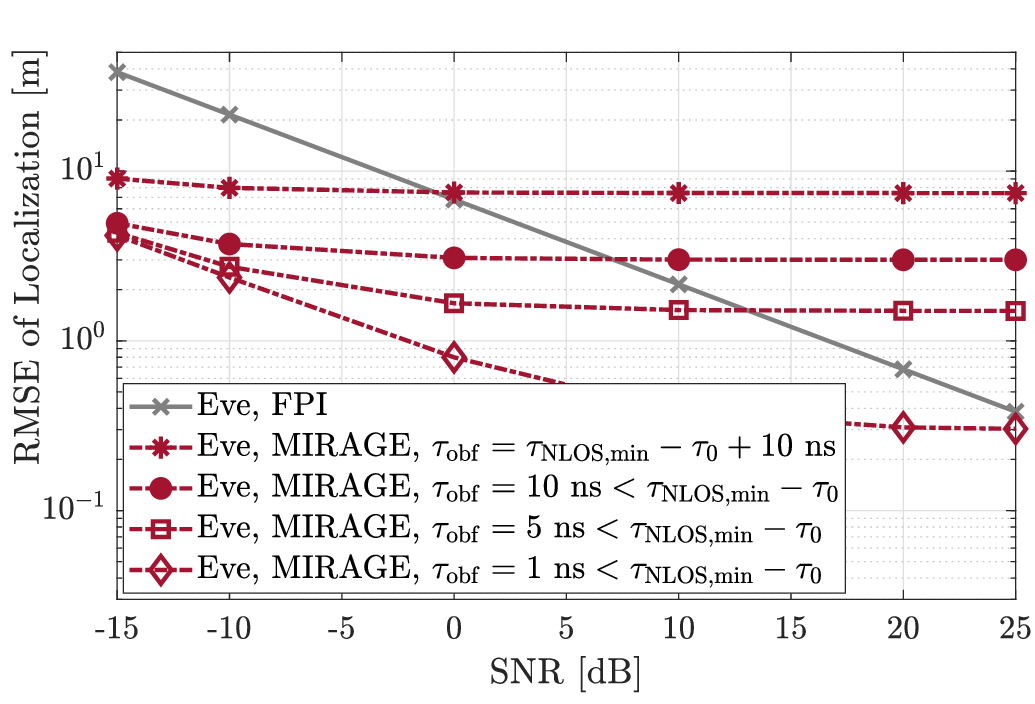}\vspace{-9pt}
\caption{\rev{Comparisons of localization accuracy: CSI-free FPI versus CSI-dependent MIRAGE \cite{Ayyalasomayajula}.}}\vspace{-16pt}%
\label{fig:MIRAGE}
\end{figure}

\rev{Furthermore, we adapt MIRAGE, a CSI-dependent beamforming design \cite{Ayyalasomayajula} for comparison to our CSI-free FPI scheme. To achieve this comparison, we compute a lower bound on Eve's localization accuracy. Since the path with the smallest TOA is typically treated as the LOS path, MIRAGE adds a delay to this path via a transmit beamformer \cite{Ayyalasomayajula} to Eve, thus obfuscating Alice's location. The MCRB \cite{Fortunatimismatchsurvey,RichmondMCRB} is evaluated to characterize the degradation of Eve's localization accuracy; we use the MCRB as Eve's perception of the channel is incorrect. We refer the reader to \cite{DAIS} for related derivations. To be concrete, we denote by $\tau_{\text{NLOS},\min}\triangleq\min_{k=1,\cdots,K} \tau_k$  the smallest TOA of all the NLOS paths,  and $\tau_{\text{obf}}$ as the delay added by the beamforming design \cite{Ayyalasomayajula}, respectively. Then, given the knowledge of the angle information, by designing the beamformer as $\bar{\bm f}_0e^{-j\frac{2\pi n\tau_{\text{obf}}}{NT_s}}+\sum_{k=1}^{K}\bar{\bm f}_k$, where $\bar{\bm f}_k$ is a vector in the null space of $[\bm\alpha(\theta_{\text{Tx},0}),\cdots,\bm\alpha(\theta_{\text{Tx},k-1}),\bm\alpha(\theta_{\text{Tx},k+1}),\cdots,\alpha(\theta_{\text{Tx},K})]^{\mathrm{H}}$, the TOA of the LOS path  perceived by Eve can be expressed as $\tau_0+\tau_{\text{obf}}$. The goal of MIRAGE is to make $\tau_0+\tau_{\text{obf}}>\tau_{\text{NLOS},\min}$ such that one of the NLOS path has the smallest TOA \cite{Ayyalasomayajula}.  Clearly the choice of $\tau_{\text{obf}}$ is dependent on the CSI between Alice and Eve.} 

\rev{
As seen in Figure \ref{fig:MIRAGE}, when $\tau_{\text{obf}}$ is set to $\tau_{\text{NLOS},\min}-\tau_0+10$ ns,  Alice's location can be effectively spoofed.  In addition, due to the geometric mismatch introduced for Eve, there is an estimation error floor with MIRAGE at high SNRs, which also indicates enhanced location privacy.  However, we reiterate that this degraded localization performance for Eve is achieved at the expense of needing the CSI between Alice and Eve.  It is not clearly how Alice would learn the delays and angles of this channel between herself and a passive eavesdropper.  We presume that there would be no coordination between Alice and Eve.  Furthermore,  MIRAGE is not robust to imperfect CSI.  To show MIRAGE's sensitivity to the accuracy of the CSI, we assume {\em perfect}\footnote{\rev{While it is impractical for Alice to have perfect angle information, computing this MCRB is non-trivial and beyond the scope of the current work.  Thus, the results show Eve's genie-aided performance and are a further lower bound on her MSE.}} knowledge of angle information, but $\tau_{\text{obf}}$ is erroneously designed due to inaccurate TOA estimates \textit{i.e.,} $\tau_0+\tau_{\text{obf}}<\tau_{\text{NLOS},\min}$. From Figure \ref{fig:MIRAGE}, when $\tau_{\text{obf}}$ is set to $10$ ns, $5$ ns, and $1$ ns, respectively, we see a clear degradation of the efficacy of MIRAGE. If we incorporated errors in the angle information, the degradation would be more significant for Eve.  Note that degradation means that Eve's localization capabilities are improved. In contrast, our scheme does not reply the CSI and can achieve a comparable accuracy degradation for Eve.}

\section{Conclusions}\label{sec:con}
A location-privacy enhancement strategy was investigated with the injection of \re{fake paths}. A novel CSI-free location-privacy enhancement framework was proposed, where the structure of the channel was exploited for the design of \re{the fake paths}. The injected \re{fake paths} were proved to be indistinguishable while two closed-form, lower bounds on the estimation error with \re{the FPI} were provided to validate that the proposed method can strongly degrade the eavesdropping ability of illegitimate devices. To effectively preserve location privacy based on the proposed framework, a transmit \revis{precoder} was designed to inject the \re{fake paths} for the reduction of the LPL to illegitimate devices, while legitimate devices can maintain localization accuracy using the securely shared information. With respect to the CRLB for localization, there was $20$dB degradation in contrast to legitimate devices with the shared information and the robustness to the leakage of \revis{precoder} structure was also highlighted. Furthermore, the efficacy of \re{the FPI} was numerically verified with the comparisons to unstructured Gaussian noise \rev{and a CSI-dependent beamforming strategy}.

\section*{\rev{Acknowledgement}}
\rev{The authors gratefully acknowledge valuable discussions with Prof. Henk Wymeersch. They also thank the anonymous reviewers for the helpful remarks.}

\appendices
\begingroup
\allowdisplaybreaks
\section{Proof of Proposition \ref{prop:iden}}\label{sec:proofiden}
From the analysis for the feasibility of paths, to show Eve cannot make a distinction between fake paths and true paths, we seek to prove all the fake paths are geometrically feasible. To find the scatterers that can produce the injected fake paths, we consider $\tilde{K}+1$ potential positions for these scatterers and denote by $b_{\tilde{k}}$ the distance between Alice and the $\tilde{k}$-th potential position, with $\tilde{k}=0,1,\cdots,\tilde{K}$. Then, the $\tilde{k}$-th potential position, denoted as $\tilde{\bm v}_{\tilde{k}}=[\tilde{v}_{{\tilde{k}},x},\tilde{v}_{{\tilde{k}},y}]^{\mathrm{T}}\in\mathbb{R}^2$, can be mapped to $\{\tilde{\tau}_{\tilde{k}}, \tilde{\theta}_{\text{Tx},{\tilde{k}}}\}$ according to Equation \eqref{eq:geometry}, {\it i.e.,} the $\tilde{k}$-th fake path is geometrically feasible, suppose that 
\begin{itemize}
    \item[C1)] the inequality 
    \begin{equation}\label{eq:condis}
        0\overset{\mathrm{(a)}}{\leq} b_{\tilde{k}}\overset{\mathrm{(b)}}{\leq} c\tilde{\tau}_{\tilde{k}} 
    \end{equation} 
    holds;
    \item[C2)] given $\tilde{\theta}_{\text{Tx},{\tilde{k}}}$, the equality 
    \begin{equation}\label{eq:condelay}
        \tilde{\tau}_{\tilde{k}} = \frac{b_{\tilde{k}} + \left\| \left[\tilde{v}_{{\tilde{k}},x}-z_x,\tilde{v}_{{\tilde{k}},y}-z_y\right]  \right\|_2}{c}, 
    \end{equation}
    is satisfied, where $\left\| \left[\tilde{v}_{{\tilde{k}},x}-z_x,\tilde{v}_{{\tilde{k}},y}-z_y\right]  \right\|_2$ represents the distance between Eve and the $\tilde{k}$-th potential position that can be expressed as 
    \begin{subequations}\label{eq:tildevk}
    \begin{align}
        \tilde{v}_{{\tilde{k}},x}= & \ b_{\tilde{k}}\cos(\tilde{\theta}_{\text{Tx},{\tilde{k}}})+p_x,
        \label{tildevkx}\\
        \tilde{v}_{{\tilde{k}},y}= & \ b_{\tilde{k}}\sin(\tilde{\theta}_{\text{Tx},{\tilde{k}}})+p_y.
        \label{tildevky}
    \end{align}
    \end{subequations}
\end{itemize}
According to Equation \eqref{eq:condelay}, we set the distance $b_{\tilde{k}}$ to 
\begin{equation}\label{eq:bk}
    b_{\tilde{k}} = \frac{\left(c\tilde{\tau}_{\tilde{k}}\right)^2-\left(z_x-p_x\right)^2-\left(z_y-p_y\right)^2}{2\left(c\tilde{\tau}_{\tilde{k}}-\left(z_x-p_x\right)\cos(\tilde{\theta}_{\text{Tx},{\tilde{k}}})-\left(z_y-p_y\right)\sin(\tilde{\theta}_{\text{Tx},{\tilde{k}}})\right)}. 
\end{equation}
It can be verified that the condition C2) is satisfied using the distance $b_{\tilde{k}}$ in Equation \eqref{eq:bk} if the condition C1) is met. Hence, the final step of this proof is to show Equation \eqref{eq:condis} holds if $b_{\tilde{k}}$ is set according to Equation \eqref{eq:bk}, under the assumption of $c\tilde{\tau}_{\tilde{k}}\geq\|\bm z- \bm p\|_2$.

For notational convenience, we introduce two variables $\dot{z}_x$ and $\dot{z}_y$ that are defined as $\dot{z}_x\triangleq z_x-p_x$ and $\dot{z}_y\triangleq z_y-p_y$. To prove $\mathrm{(a)}$, considering the assumption $c\tilde{\tau}_{\tilde{k}}\geq\|\bm z- \bm p\|_2>0$, we equivalently need to justify 
\begin{equation}
    c\tilde{\tau}_{\tilde{k}}-\dot{z}_x\cos(\tilde{\theta}_{\text{Tx},{\tilde{k}}})-\dot{z}_y\sin(\tilde{\theta}_{\text{Tx},{\tilde{k}}})\geq0,\label{eq:deinproof} 
\end{equation}
which holds since we have  
\begin{equation}
\begin{aligned}
&\left(\dot{z}_x\cos(\tilde{\theta}_{\text{Tx},{\tilde{k}}})+\dot{z}_y\sin(\tilde{\theta}_{\text{Tx},{\tilde{k}}})\right)^2= \\
& \dot{z}_x^2\cos^2(\tilde{\theta}_{\text{Tx},{\tilde{k}}})+\dot{z}_y^2\sin^2(\tilde{\theta}_{\text{Tx},{\tilde{k}}})+2\dot{z}_x\dot{z}_y\sin(\tilde{\theta}_{\text{Tx},{\tilde{k}}})\cos(\tilde{\theta}_{\text{Tx},{\tilde{k}}})\\
&\overset{\mathrm{(c)}}{=} \dot{z}_x^2+\dot{z}_y^2-\left(\dot{z}_x\sin(\tilde{\theta}_{\text{Tx},{\tilde{k}}})-\dot{z}_y\cos(\tilde{\theta}_{\text{Tx},{\tilde{k}}})\right)^2\\
&\leq \dot{z}_x^2+\dot{z}_y^2 \overset{\mathrm{(d)}}{\leq} \left(c\tilde{\tau}_{\tilde{k}}\right)^2.\label{eq:ineqinproof} 
\end{aligned}
\end{equation}
The equality $\mathrm{(c)}$ can be verified with the property of the trigonometric functions while inequality $\mathrm{(d)}$ holds according to the assumption $c\tilde{\tau}_{\tilde{k}}\geq\|\bm z- \bm p\|_2$. In addition, we can re-express $c\tilde{\tau}_{\tilde{k}}-b_{\tilde{k}}$ as 
\begin{equation}
\frac{\left(c\tilde{\tau}_{\tilde{k}}\right)^2+\dot{z}_x^2+\dot{z}_y^2-2c\tilde{\tau}_{\tilde{k}}\left(\dot{z}_x\cos(\tilde{\theta}_{\text{Tx},{\tilde{k}}})+\dot{z}_y\sin(\tilde{\theta}_{\text{Tx},{\tilde{k}}})\right)}{2\left(c\tilde{\tau}_{\tilde{k}}-\dot{z}_x\cos(\tilde{\theta}_{\text{Tx},{\tilde{k}}})-\dot{z}_y\sin(\tilde{\theta}_{\text{Tx},{\tilde{k}}})\right)}, 
\end{equation}
where the denominator has been proved to be non-negative according to Equation \eqref{eq:deinproof}. Hence, to validate $\mathrm{(b)}$, we need to prove that 
\begin{equation}
\left(c\tilde{\tau}_{\tilde{k}}\right)^2+\dot{z}_x^2+\dot{z}_y^2-2c\tilde{\tau}_{\tilde{k}}\left(\dot{z}_x\cos(\tilde{\theta}_{\text{Tx},{\tilde{k}}})+\dot{z}_y\sin(\tilde{\theta}_{\text{Tx},{\tilde{k}}})\right)\geq0. \label{eq:noinproof} 
\end{equation}
Since we have 
\begin{equation}
\begin{aligned}
&\left(2c\tilde{\tau}_{\tilde{k}}\left(\dot{z}_x\cos(\tilde{\theta}_{\text{Tx},{\tilde{k}}})+\dot{z}_y\sin(\tilde{\theta}_{\text{Tx},{\tilde{k}}})\right)\right)^2\\
&\overset{\mathrm{(e)}}{\leq} 4\left(c\tilde{\tau}_{\tilde{k}}\right)^2 (\dot{z}_x^2+\dot{z}_y^2)\leq\left(\left(c\tilde{\tau}_{\tilde{k}}\right)^2+\dot{z}_x^2+\dot{z}_y^2\right)^2, 
\end{aligned}
\end{equation}
where $\mathrm{(e)}$ follows from Equation \eqref{eq:ineqinproof}, Equation \eqref{eq:noinproof} holds, concluding the proof.

\section{Proof of Lemma \ref{lemma:converge}}\label{sec:proofconverge}
$ 
$
Proving Lemma \ref{lemma:converge} is equivalent to showing that each element of $\frac{1}{G}{\bm J}^{(\bm \zeta_k)}$ converges a.s. to the corresponding element of $\breve{\bm J}^{(\bm \zeta_k)}$ as $G\rightarrow\infty$. Therefore, according to their definitions with $\frac{\partial u^{(g,n)}}{\partial \bm{\bar\tau}[k]}$ and $\frac{\partial u^{(g,n)}}{\partial \bm{\bar\theta}_{\text{Tx}}[k]}$ presented in Equation \eqref{eq:partial}, we restrict the derivations to 
\begin{equation}
\begin{aligned}
    &\frac{2}{G\sigma^2}\sum_{n=0}^{N-1}\sum_{g=1}^{G} \mathfrak{R}\left\{\left(\frac{\partial  u^{(g,n)}}{\partial \bm{\bar\tau}[k]}\right)^{\mathrm{H}}\frac{\partial  u^{(g,n)}}{\partial \bm{\bar\tau}[k]}\right\}\\
    &\quad\quad\quad\quad\quad\quad\overset{\text{a.s.}}{\longrightarrow}\frac{8\pi^2 O_1O_6}{(\sigma{NT_s})^2}|\bm{\bar\gamma}[k]|^2, \text{ as } G\rightarrow\infty, 
\end{aligned}
\end{equation}
in Equation \eqref{eq:convproof} as the others can be verified analogously, where $\mathrm{(f)}$ results from the law of large number. 
\begin{figure*}[b] 
\vspace{-5pt}
\hrulefill
\vspace{-3pt}
    \begin{align} \label{eq:convproof}
        &\frac{2}{G\sigma^2}\sum_{n=0}^{N-1}\sum_{g=1}^{G} \mathfrak{R}\left\{\left(\frac{\partial  u^{(g,n)}}{\partial \bm{\bar\tau}[k]}\right)^{\mathrm{H}}\frac{\partial  u^{(g,n)}}{\partial \bm{\bar\tau}[k]}\right\}\nonumber\\
        &=\frac{2}{G\sigma^2}\sum_{n=0}^{N-1}\sum_{g=1}^{G} \mathfrak{R}\left\{\left(j\frac{2\pi \sqrt{N_t} n}{NT_s}\bar\gamma^{\mathrm{H}}_k e^{j\frac{2\pi n\bar\tau_k}{NT_s}}\left(\bm s^{(g,n)}\right)^{\mathrm{H}}\bm\alpha(\bar\theta_{\text{Tx},k}) \right)\left(-j\frac{2\pi \sqrt{N_t} n}{NT_s}\bar\gamma_k e^{-j\frac{2\pi n\bar\tau_k}{NT_s}}\bm\alpha(\bar\theta_{\text{Tx},k})^{\mathrm{H}}\bm s^{(g,n)}\right)\right\}\nonumber\\
        &=\frac{2}{\sigma^2}\sum_{n=0}^{N-1}\mathfrak{R}\left\{\left(\frac{2\pi \sqrt{N_t} n}{NT_s}\right)^2|\bar\gamma_k|^2\bm\alpha(\bar\theta_{\text{Tx},k})^{\mathrm{H}}\left(\frac{1}{G}\sum_{g=1}^{G}\left(\bm s^{(g,n)}\right)\left(\bm s^{(g,n)}\right)^{\mathrm{H}}\right)\bm\alpha(\bar\theta_{\text{Tx},k})\right\}\nonumber\\
        &\xrightarrow[\text{a.s.}]{\mathrm{(f)}} \frac{2}{\sigma^2}\sum_{n=0}^{N-1}\mathfrak{R}\left\{\left(\frac{2\pi  n}{NT_s}\right)^2|\bar\gamma_k|^2\bm\alpha(\bar\theta_{\text{Tx},k})^{\mathrm{H}}\bm\alpha(\bar\theta_{\text{Tx},k})\right\} =\frac{2}{\sigma^2}\sum_{n=0}^{N-1}\left(\frac{2\pi  n}{NT_s}\right)^2|\bm{\bar\gamma}[k]|^2
        = \frac{8\pi^2  O_1O_6}{(\sigma{NT_s})^2}|\bm{\bar\gamma}[k]|^2.
    \end{align}

\end{figure*}

\section{Proof of Proposition \ref{prop:errorbound}}\label{sec:prooferrorbound}
For the $k$-th true path and the $k$-th fake path, considering the TOAs and AODs of the other paths as well as all the channel coefficients as the nuisance parameters, we can bound the mean squared error of $\rev{\hat{\bm \zeta}_{k}}$ as $\mathbb{E}\left\{\left(\rev{\hat{\bm \zeta}_{k}}-\rev{{\bm \zeta}_{k}}\right)^{\mathrm{T}}\left(\rev{\hat{\bm \zeta}_{k}}-\rev{{\bm \zeta}_{k}}\right)\right\} \geq  \operatorname{Tr}\left({\left(\bm J^{(\bm\zeta_k)}\right)^{-1}}\right)$ according to \cite{Scharf}. Following from Lemma \ref{lemma:converge}, if $G\rightarrow\infty$, we have $G\operatorname{Tr}\left({\left(\bm J^{(\bm\zeta_k)}\right)^{-1}}\right)\overset{\text{a.s.}}{\longrightarrow}\operatorname{Tr}\left({\left(\breve{\bm J}^{(\bm\zeta_k)}\right)^{-1}}\right)$. Hence, for any real $\psi>0$, there always exists a positive integer $\mathcal{G}$ such that when $G\geq\mathcal{G}$, we have $G\operatorname{Tr}\left({\left(\bm J^{(\bm\zeta_k)}\right)^{-1}}\right)>\operatorname{Tr}\left({\left(\breve{\bm J}^{(\bm\zeta_k)}\right)^{-1}}\right)-\psi$ with probability of 1. Then, the final step is prove $\operatorname{Tr}\left({\left(\breve{\bm J}^{(\bm\zeta_k)}\right)^{-1}}\right)\geq\Xi^{(k)}$.

Denoting by $\rho^{(k)}_i$ the $i$-th eigenvalue of $\breve{\bm J}^{(\bm\zeta_k)}$ with $i=1,2,3,4$,
we have the following lower bound on $\operatorname{Tr}\left({\left(\breve{\bm J}^{(\bm\zeta_k)}\right)^{-1}}\right)$, 
\begin{equation}
\operatorname{Tr}\left({\left(\breve{\bm J}^{(\bm\zeta_k)}\right)^{-1}}\right) = \sum_{i=1}^4\frac{1}{\rho^{(k)}_i}\overset{\mathrm{(g)}}{\geq} \Xi^{(k)}\triangleq4\sqrt[4]{\frac{1}{\prod_{i=1}^4\rho^{(k)}_i}}, 
\end{equation}
where $\mathrm{(g)}$ follows from the inequality of arithmetic and geometric means.
Then, leveraging Equation \eqref{eq:L}, the equality that $\prod_{i=1}^4\rho^{(k)}_i=\operatorname{det}\left(\breve{\bm J}^{(\bm\zeta_k)}\right)$, and the assumption of $\delta_{\gamma_k}=\delta_{\gamma_k,\min}=0$ yields the desired statement.

\vspace{-9pt}
\section{Proof of Corollary \ref{coro:errorbound}}\label{sec:proofserrorbound}

Under the assumption A1, the lower bound with $\Xi^{(k)}$ has been derived in Proposition \ref{prop:errorbound}. Given the assumptions A2 and A3, Equation \eqref{eq:boundonMiMj} holds so we have the following inequalities, \vspace{-10pt}
\begin{subequations}\label{eq:boundonMiMjOiOj}
    \begin{align}
        2O_1O_6-\epsilon<\mathfrak{R}\left\{M_1^{(k)}M_6^{(k)}\right\}+O_1O_6,\\
        2O_4O_5-\epsilon<\mathfrak{R}\left\{M_4^{(k)}M_5^{(k)}\right\}+O_4O_5\\
        \mathfrak{R}\left\{M_2^{(k)}M_3^{(k)}\right\}+O_2O_3<2O_2O_3+\epsilon. \vspace{-3pt}
    \end{align}
\end{subequations}
\begin{figure*}[b]
\vspace{-5pt}
\hrulefill
\vspace{-3pt}
\begin{equation}\label{eq:Xi1pos}
\begin{aligned}
    &\left(\mathfrak{R}\left\{M_1^{(k)}M_6^{(k)}\right\}+O_1O_6\right) \left(\mathfrak{R}\left\{M_4^{(k)}M_5^{(k)}\right\}+O_4O_5\right)-\left(\mathfrak{R}\left\{M_2^{(k)}M_3^{(k)}\right\}+O_2O_3\right)^2\\
    >&\left(2O_1O_6-\epsilon\right) \left(2O_4O_5-\epsilon\right)-\left(2O_2O_3+\epsilon\right)^2 = 4\left(O_1O_6O_4O_5-\left(O_2O_3\right)^2\right)-2\epsilon(O_1O_6+2O_2O_3+O_4O_5)=0.
    \end{aligned}
\end{equation}
\end{figure*}With the definition of $\epsilon$, $\Xi^{(k)}_1$ in Equation \eqref{eq:lowerboundXi} is a positive value according to Equation \eqref{eq:Xi1pos} so 
\begin{equation}
\begin{aligned}\label{eq:lowerboundXi1}
    \Xi^{(k)}_1\geq&\frac{(NT_s)^4}{\left(O_1O_6+\mathfrak{R}\left\{M_1M_6\right\}\right)\left(O_4O_5+\mathfrak{R}\left\{M_4M_5\right\}\right)}\\
    \overset{\mathrm{(h)}}{\geq}&\frac{(NT_s)^4}{4O_1O_4O_5O_6} 
    \end{aligned}
\end{equation}
holds, where $\mathrm{(h)}$ follows from Equation \eqref{eq:boundM}. In addition, it can be verified that $\breve{\bm J}^{(\bm\zeta_k)}$ is a positive semidefinite matrix so $\operatorname{det}\left(\breve{\bm J}^{(\bm\zeta_k)}\right)\geq 0$ and thus $\Xi^{(k)}_2\geq0$ according to Equation \eqref{eq:Xi1pos}. Then, we have 
\begin{equation}\label{eq:lowerboundXi2}
    \Xi^{(k)}_2\geq\frac{1}{\left(O_1O_6-\mathfrak{R}\left\{M_1M_6\right\}\right)\left(O_4O_5-\mathfrak{R}\left\{M_4M_5\right\}\right)},
\end{equation}
where the denominator of the right-hand side can be bounded according to \vspace{-3pt}
\begin{equation}\label{eq:upperboundO1O6}
    \begin{aligned}
    &O_1O_6-\mathfrak{R}\left\{M_1M_6\right\}\\
    =&\mathfrak{R}\left\{O_1O_6-M_1M_6\right\}\\
    \leq&|O_1O_6-M_1M_6|\\
    =&\left|\frac{1}{N_t}\sum_{n=0}^{N-1}\sum_{n_t=0}^{N_t-1}n^2\left(1-e^{j2\pi\left(\frac{ n_t d \sin\left(\delta_{\theta_{\text{Tx},k}}\right)}{\lambda_c}-\frac{n\delta_{\tau_k}}{NT_s}\right)}\right)\right|\\
    \overset{\mathrm{(i)}}{\leq}&\frac{1}{N_t}\sum_{n=0}^{N-1}\sum_{n_t=0}^{N_t-1}n^2\left|\left(1-e^{j2\pi\left(\frac{ n_t d \sin\left(\delta_{\theta_{\text{Tx},k}}\right)}{\lambda_c}-\frac{n\delta_{\tau_k}}{NT_s}\right)}\right)\right|\\
    \overset{\mathrm{(j)}}{\leq}&\frac{3\pi}{N_t}\sum_{n=0}^{N-1}\sum_{n_t=0}^{N_t-1}n^2\left|\frac{ n_t d \sin\left(\delta_{\theta_{\text{Tx},k}}\right)}{\lambda_c}-\frac{n\delta_{\tau_k}}{NT_s}\right|\\
    \overset{\mathrm{(k)}}{=}&\frac{3\pi}{N_t}\sum_{n=0}^{N-1}\left(\frac{n^3\delta_{\tau_k}}{NT_s}+\sum_{n_t=1}^{N_t-1}n^2\left(\frac{ n_t d \sin\left(\delta_{\theta_{\text{Tx},k}}\right)}{\lambda_c}-\frac{n\delta_{\tau_k}}{NT_s}\right)\right)\\
    =&\frac{\pi N(N-1)(2N-1)(N_t-1)d\sin\left(\delta_{\theta_{\text{Tx},k}}\right)}{4\lambda_c}\\
    & -\frac{3\pi N(N-1)^2(N_t-2)\delta_{\tau_k}}{4N_tT_s}, \vspace{-5pt}
    \end{aligned}
\end{equation}
and \vspace{-3pt}
\begin{equation}\label{eq:upperboundO4O5}
    \begin{aligned}
    O_4O_5-\mathfrak{R}\left\{M_4M_5\right\}\leq&\frac{3\pi NN_t(N_t-1)^2d\sin\left(\delta_{\theta_{\text{Tx},k}}\right)}{4\lambda_c} \\&-\frac{\pi (N_t-1)(2N_t-1)(N-1)\delta_{\tau_k}}{4T_s}. \vspace{-3pt}
    \end{aligned}
\end{equation}
Herein, it can be verified with some algebra that ${\mathrm{(i)}}$, ${\mathrm{(j)}}$, and ${\mathrm{(k)}}$ result from the triangle inequality, the assumption A4, and the assumption A5, respectively, while Equation \eqref{eq:upperboundO4O5} can be derived analogously. By leveraging Equations \eqref{eq:lowerboundXi2}, \eqref{eq:upperboundO1O6} and \eqref{eq:upperboundO4O5}, the quantity $\Xi^{(k)}_2$ can be further bounded as provided in Equation \eqref{eq:flowerboundXi2}, where  ${\mathrm{(l)}}$ follows from the fact that $N,N_t\geq 2$ holds for MISO OFDM systems. Then, substituting Equations \eqref{eq:lowerboundXi1} and \eqref{eq:flowerboundXi2} into Equation \eqref{eq:lowerboundXi} simplifies $\Xi^{(k)}$ into $\aleph^{(k)}$ shown in Equation \eqref{eq:lowerboundaleph}. Since the lower bound $\aleph^{(k)}$ holds for any $\delta_{\tau_k}$ and $\delta_{\theta_{\text{Tx},k}}$ that satisfy assumptions A1-A5, we can set $ \delta_{\tau_k}=\frac{\sin\left(\delta_{\theta_{\mathrm{Tx},{k}}}\right)}{(N-1)\Lambda}$ for $\aleph^{(k)}$ so that Equation \eqref{eq:slowerboundXi} can be verified with some algebra.
\begin{figure*}[b]
\vspace{-5pt}
\hrulefill
\vspace{-3pt}
\begin{equation}\label{eq:flowerboundXi2}
\begin{aligned}
    \Xi_2^{(k)}\geq&\frac{4 N_tT_s\Lambda}{\pi(N-1)\left((2N-1)(N_t-1)N_t\sin\left(\delta_{\theta_{\text{Tx},k}}\right)-3\Lambda N(N-1)(N_t-2)\delta_{\tau_k}\right)}\\
    &\times\frac{4 T_s\Lambda}{\pi(N_t-1)\left(3N_t(N_t-1)\sin\left(\delta_{\theta_{\text{Tx},k}}\right)-\Lambda(2N_t-1)(N-1)\delta_{\tau_k}\right)}\\
    \overset{\mathrm{(l)}}{\geq}&\frac{(4 T_s\Lambda)^2}{\pi^2N_tN^2\left((2N_t^2\sin\left(\delta_{\theta_{\text{Tx},k}}\right)-3\Lambda(N-1)(N_t-2)\delta_{\tau_k}\right)\left(3N_t\sin\left(\delta_{\theta_{\text{Tx},k}}\right)-\Lambda(N-1)\delta_{\tau_k}\right)}\vspace{-5pt}
    \end{aligned}
\end{equation}
\vspace{-3pt}
\hrulefill
\begin{equation}\label{eq:lowerboundaleph}
    \begin{aligned}
    \aleph^{(k)}=&\frac{\lambda_c\sigma^2}{\sqrt{2}\pi^\frac{5}{2}d|\gamma_k|^2|\cos(\theta_{\text{Tx},k})\cos(\tilde{\theta}_{\text{Tx},k})|}\\
        &\times\sqrt[4]{\frac{(N\Lambda)^2T_s^6}{O_1O_4O_5O_6N_t\left(2N_t^2\sin\left(\delta_{\theta_{\text{Tx},k}}\right)-3\Lambda(N-1)(N_t-2)\delta_{\tau_k}\right)\left(3N_t\sin\left(\delta_{\theta_{\text{Tx},k}}\right)-\Lambda(N-1)\delta_{\tau_k}\right)}}
    \end{aligned}
    \end{equation}
\end{figure*} 

\vspace{-8pt}
\section{Proof of Proposition \ref{prop:singularFIMbeamformer}}\label{sec:proofsingularFIMbeamformer}
Since we assume that Eve knows Alice's \revis{precoder} structure, the associated noise-free observation is defined as   \vspace{-3pt}
\begin{equation}
     \iota^{(g,n)} \triangleq \sqrt{N_t} \sum_{k=0}^{K}{\gamma}_ke^{-j\frac{2\pi n{\tau}_k}{NT_s}}\bm\alpha({\theta}_{\text{Tx},k})^{\mathrm{H}} \tilde{\bm s}^{(g,n)}. \vspace{-3pt}
\end{equation}
Let $ \varpi^{(g,n)}\triangleq\sqrt{N_t} \sum_{k=0}^{K}{\gamma}_ke^{-j\frac{2\pi n{\tau}_k}{NT_s}}\bm\alpha({\theta}_{\text{Tx},k})^{\mathrm{H}} {\bm s}^{(g,n)}$. It can be verified that $ \iota^{(g,n)}\rightarrow 2\varpi^{(g,n)}$ as $\bar\delta_{\tau}, \bar\delta_{\theta_{\mathrm{Tx}}}\rightarrow 0$.
To derive the FIM for Eve’s channel estimation with the prior knowledge of the \revis{precoder} structure, we can compute the derivative $\frac{\partial  \iota^{(g,n)}}{\partial \xi_k}$ 
by replacing ${\bm s}^{(g,n)}$ in Equation \eqref{eq:partial} with $\tilde{\bm s}^{(g,n)}$, where $\xi_k\in\{\tau_k,\theta_{\mathrm{Tx},k},\mathfrak{R}\{{\gamma_k}\},\mathfrak{I}\{{\gamma_k}\}\}$, while derive $\frac{\partial  \iota^{(g,n)}}{\partial \bar\delta_{\tau}}$ and $\frac{\partial  \iota^{(g,n)}}{\partial \bar\delta_{{\theta}_{\text{Tx}}}}$ as follows, \vspace{-5pt}
\begin{subequations}\label{eq:partialbeamformer}
\begin{align}
    \frac{\partial \iota^{(g,n)}}{\partial \bar\delta_{\tau}} &= -j\frac{2\pi {N_t} n}{NT_s}\sum_{k=0}^{K}{\gamma}_k e^{-j\frac{2\pi n \left({\tau}_k+\delta_{\tau}\right)}{NT_s}}{\bm\alpha({\theta}_{\text{Tx},k})^{\mathrm{H}}}\nonumber\\
    &\quad\times \operatorname{diag}\left(\bm \alpha\left(\bar{\delta}_{\theta_\text{Tx}}\right)^{\mathrm{H}} \right)\bm s^{(g,n)},\\
    \frac{\partial \iota^{(g,n)}}{\partial \bar\delta_{{\theta}_{\text{Tx}}}} &=  j\frac{2\pi {N_t} d}{\lambda_c}\sum_{k=0}^{K}{\gamma}_k e^{-j\frac{2\pi n \left({\tau}_k+\delta_{\tau}\right)}{NT_s}}\cos(\bar\delta_{{\theta}_{\text{Tx}}})\nonumber\\
    &\quad\times {\bm\alpha({\theta}_{\text{Tx},k})^{\mathrm{H}}}\operatorname{diag}([0,1,\cdots,N_t-1])\nonumber\\
    &\quad\times\operatorname{diag}\left(\bm \alpha\left(\bar{\delta}_{\theta_\text{Tx}}\right)^{\mathrm{H}} \right)\bm s^{(g,n)}. \vspace{-5pt}
\end{align} 
\end{subequations}
Then, it is straightforward to show that, as 
$\bar\delta_{\tau}, \bar\delta_{\theta_{\mathrm{Tx}}}\rightarrow 0$, \vspace{-4pt}
\begin{subequations}\label{eq:partialbeamformer}
\begin{align}
    \frac{\partial  \iota^{(g,n)}}{\partial \xi_k}&\rightarrow  2\frac{\partial \varpi^{(g,n)}}{\partial \xi_k},\\
    \frac{\partial  \iota^{(g,n)}}{\partial \bar\delta_{\tau}}&\rightarrow \sum_{k=0}^{K} \frac{\partial \varpi^{(g,n)}}{\partial {\tau}_k}, \\
    \frac{\partial  \iota^{(g,n)}}{\partial \bar{\delta}_{\theta_\text{Tx}}}&\rightarrow \sum_{k=0}^{K} \frac{\partial \varpi^{(g,n)}}{\partial {\theta}_{\text{Tx},k}}. \vspace{-5pt}
\end{align} 
\end{subequations}
On the other hand, similar to Equation \eqref{eq:convproof}, it can be verified that for any $\xi_k$ and $\xi_{k^{\prime}}\in\{\tau_{k^{\prime}},\theta_{\mathrm{Tx},k^{\prime}},\mathfrak{R}\{{\gamma_{k^{\prime}}}\},\mathfrak{I}\{{\gamma_{k^{\prime}}}\}\}$, \vspace{-4pt}
\begin{equation}
\begin{aligned}
    \frac{2}{\sigma^2}\sum_{n=0}^{N-1}\sum_{g=1}^{G}\mathfrak{R}\left\{\left(\frac{\partial\varpi^{(g,n)}}{\partial \xi_k}\right)^{\mathrm{H}}\frac{\partial \varpi^{(g,n)}}{\partial \xi_{k^{\prime}}}\right\}\overset{\text{a.s.}}{\longrightarrow} 0, \ k\neq {k^{\prime}} , \vspace{-3pt}
\end{aligned}
\end{equation}
holds when $G,N_t\rightarrow\infty$. Hence, given $\bar\delta_{\tau}, \bar\delta_{\theta_{\mathrm{Tx}}}\rightarrow 0$ and $G,N_t\rightarrow\infty$, for any $\xi_k$, we have  \vspace{-3pt}
\begin{subequations}
\begin{align}
    2\bm J^{(\bm \chi)}_{\bar\delta_{\tau},\xi_k}=2\bm J^{(\bm \chi)}_{\xi_k,\bar\delta_{\tau}}&\rightarrow \bm J^{(\bm \chi)}_{{\tau}_k,\xi_k}, \\
    2\bm J^{(\bm \chi)}_{\bar{\delta}_{\theta_\text{Tx}},\xi_k}=2\bm J^{(\bm \chi)}_{\xi_k,\bar{\delta}_{\theta_\text{Tx}}}&\rightarrow \bm J^{(\bm \chi)}_{{\theta}_{\text{Tx},k},\xi_k},   \\
    4\bm J^{(\bm \chi)}_{\bar\delta_{\tau},\bar\delta_{\tau}}&\rightarrow \sum_{k=0}^{K} \bm J^{(\bm \chi)}_{{\tau}_k,{\tau}_k},\\
    4\bm J^{(\bm \chi)}_{\bar{\delta}_{\theta_\text{Tx}},\bar{\delta}_{\theta_\text{Tx}}}&\rightarrow \sum_{k=0}^{K} \bm J^{(\bm \chi)}_{{\theta}_{\text{Tx},k},{\theta}_{\text{Tx},k}},\\
    4\bm J^{(\bm \chi)}_{\bar\delta_{\tau},\bar{\delta}_{\theta_\text{Tx}}}=4\bm J^{(\bm \chi)}_{\bar{\delta}_{\theta_\text{Tx}},\bar\delta_{\tau}}&\rightarrow \sum_{k=0}^{K} \bm J^{(\bm \chi)}_{{\tau}_k,{\theta}_{\text{Tx},k}}=\sum_{k=0}^{K} \bm J^{(\bm \chi)}_{{\theta}_{\text{Tx},k},{\tau}_k}. \vspace{-3pt}
\end{align} 
\end{subequations}
Hence, two rows of $\bm J^{(\bm \chi)}$ are linearly dependent on the other rows as $\bar\delta_{\tau}, \bar\delta_{\theta_{\mathrm{Tx}}}\rightarrow 0$ and $G,N_t\rightarrow\infty$, {\it i.e.,} \vspace{-3pt}
\begin{subequations}
\begin{align}
        &2\left[\bm J^{(\bm \chi)}_{\bar\delta_{\tau},\tau_0}, \bm J^{(\bm \chi)}_{\bar\delta_{\tau},\theta_{\text{Tx},0}},\cdots,\bm J^{(\bm \chi)}_{\bar\delta_{\tau},\bar{\delta}_{\theta_\text{Tx}}}\right]\nonumber\\
        \rightarrow&\sum_{k=0}^{K} \left[\bm J^{(\bm \chi)}_{\tau_k,\tau_0}, \bm J^{(\bm \chi)}_{\tau_k,\theta_{\text{Tx},0}},\cdots,\bm J^{(\bm \chi)}_{\tau_k,\bar{\delta}_{\theta_\text{Tx}}}\right],\\
        &2\left[\bm J^{(\bm \chi)}_{\bar{\delta}_{\theta_\text{Tx}},\tau_0}, \bm J^{(\bm \chi)}_{\bar{\delta}_{\theta_\text{Tx}},\theta_{\text{Tx},0}},\cdots,\bm J^{(\bm \chi)}_{\bar{\delta}_{\theta_\text{Tx}},\bar{\delta}_{\theta_\text{Tx}}}\right]\nonumber\\
        \rightarrow&\sum_{k=0}^{K} \left[\bm J^{(\bm \chi)}_{\theta_{\text{Tx},k},\tau_0}, \bm J^{(\bm \chi)}_{\theta_{\text{Tx},k},\theta_{\text{Tx},0}},\cdots,\bm J^{(\bm \chi)}_{\theta_{\text{Tx},k},\bar{\delta}_{\theta_\text{Tx}}}\right], \vspace{-5pt}
\end{align} 
\end{subequations}
which leads to the desired statement.


\renewcommand*{\bibfont}{\footnotesize}
\printbibliography

\end{document}